\documentclass[manuscript,screen,acmsmall,nonacm]{acmart}
\acmJournal{TEAC}
\AtBeginDocument{%
  }

\setcopyright{acmlicensed}
\copyrightyear{2026}
\acmYear{2026}
\acmDOI{XXXXXXX.XXXXXXX}
\thanks{A preliminary version of this paper appeared in the Proceedings of SAGT 2024 \citep{glitzner24sagt} and the full version in the ACM Transactions on Economics and Computation \citep{structuralteac}.}

\usepackage{algorithm}
\usepackage{algpseudocode}
\usepackage{hyperref}

\usepackage{pgfplots}
\pgfplotsset{compat=1.18}

\definecolor{acmDarkBlue}{RGB}{0,114,178}
\definecolor{acmGreen}{RGB}{0,158,115}
\definecolor{acmPink}{RGB}{204,121,167}
\definecolor{acmOrange}{RGB}{213,94,0}
\definecolor{acmYellow}{RGB}{240,228,66}
\definecolor{acmLightBlue}{RGB}{86,180,233}

\pgfplotsset{
    cycle list={
        {acmDarkBlue, mark=*},
        {acmGreen, mark=square*},
        {acmPink, mark=triangle*},
        {acmOrange, mark=diamond*},
        {acmYellow, mark=o},
        {acmLightBlue, mark=star}
    }
}

\newcommand*\circled[1]{\tikz[baseline=(char.base)]{
            \node[shape=circle,draw,inner sep=1.5pt] (char) {#1};}}
\newcommand*\dottedcircled[1]{\tikz[baseline=(char.base)]{
            \node[shape=circle,draw,dotted,inner sep=1.5pt] (char) {#1};}}

\newcommand\doubleplus{+\kern-2ex+\kern0.8ex}

\usetikzlibrary{arrows,fit,positioning,shapes,chains,fit,shapes,calc,backgrounds}

\usepackage[table]{xcolor}
\usepackage{soul,colortbl}

\begin{document}

\title{Structural and Algorithmic Results for Stable Cycles and Partitions in the Roommates Problem}

\author{Frederik Glitzner}
\email{f.glitzner.1@research.gla.ac.uk}
\orcid{0009-0002-2815-6368}
\author{David Manlove}
\email{david.manlove@glasgow.ac.uk}
\orcid{0000-0001-6754-7308}
\affiliation{%
  \institution{University of Glasgow}
  \city{Glasgow}
  \country{United Kingdom}
}

\renewcommand{\shortauthors}{F. Glitzner and D. Manlove}

\begin{abstract}
    In the {\sc Stable Roommates} problem, we seek a \emph{stable matching} of the agents into pairs, in which no two agents have an incentive to deviate from their assignment. It is well known that a stable matching is unlikely to exist, but a \emph{stable partition} always does and provides a succinct certificate for the \emph{unsolvability} of an instance. Furthermore, apart from being a useful structural tool to study the problem, every stable partition corresponds to a \emph{stable half-matching}, which has applications, for example, in sports scheduling and time-sharing. 

    We establish new structural results for stable partitions and show how to enumerate all stable partitions and the cycles included in such structures efficiently. We also adapt optimality criteria from stable matchings to stable partitions and give complexity and approximability results for the problems of computing such ``fair'' and ``optimal'' stable partitions. 
    
    Through this research, we contribute to a deeper understanding of stable partitions from a combinatorial point of view, as well as the computational complexity of computing ``fair" or ``optimal" stable half-matchings in practice, closing the gap between integral and fractional stable matchings and paving the way for further applications of stable partitions to unsolvable instances and computationally hard stable matching problems.
\end{abstract}




%\received{22 November 2024}
%\received[revised]{17 August 2025}
%\received[accepted]{29 December 2025}

\maketitle

\section{Introduction}
\label{section:introduction}

\subsection{Background}
\label{section:background}

The {\sc Stable Roommates} problem ({\sc sr}) is a classical combinatorial problem with applications to computational social choice. Consider a group of friends that want to play one hour of tennis, where everyone has ordinal preferences over who to play with. Can we match them into pairs such that no two friends prefer to play with each other rather than their assigned partners? If the problem instance $I$ admits such a \emph{stable matching} $M$ which does not admit a \emph{blocking pair} of agents which would rather be matched to each other than their partners in $M$, then we call $I$ \emph{solvable}. Otherwise, we call $I$ \emph{unsolvable}. Even an instance with as few as 4 agents may be unsolvable, as \citet{gale_shapley} showed. A solvable instance is shown in Example \ref{table:solvable}, where $a_1$ ranks the agents $a_2$, $a_5$, $a_3$, etc., in linear order. The problem instance admits the matching $M=\{\{a_1, a_2\}, \{a_3, a_4\}, \{a_5,a_6\} \}$ indicated in circles, the stability of which can be easily verified by hand by looking at all agents that are not matched to their first choice and seeing that no agent preferred to their partner would rather be matched to them than their partner. On the other hand, Example \ref{table:unsolvable} shows an example instance with six agents that does not admit any stable matching (for which we will see a certificate later).

\begin{table}[!htb]
\centering
\begin{minipage}{.45\linewidth}
\centering
\caption{A solvable {\sc sr} instance}
    \begin{tabular}{ c | c c c c c }
    $a_1$ & \circled{$a_2$} & $a_5$ &  $a_3$ & $a_4$ & $a_6$ \\
    $a_2$ & $a_5$ & $a_3$ & \circled{$a_1$} & $a_6$ & $a_4$ \\
    $a_3$ & \circled{$a_4$} & $a_2$ & $a_6$ & $a_1$ & $a_5$ \\
    $a_4$ & $a_1$ & $a_2$ & $a_5$ & \circled{$a_3$} & $a_6$ \\
    $a_5$ & \circled{$a_6$} & $a_2$ & $a_1$ & $a_4$ & $a_3$ \\
    $a_6$ & \circled{$a_5$} & $a_3$ & $a_2$ & $a_4$ & $a_1$
    \end{tabular}
\label{table:solvable}
\end{minipage}%
\begin{minipage}{.45\linewidth}
\centering
\caption{An unsolvable {\sc sr} instance}
    \begin{tabular}{ c | c c c c c }
    $a_1$ & \circled{$a_2$} & $a_5$ & \dottedcircled{$a_3$} & $a_4$ & $a_6$ \\
    $a_2$ & $a_4$ & \circled{$a_3$} & \dottedcircled{$a_1$} & $a_6$ & $a_5$ \\
    $a_3$ & $a_5$ & $a_4$ & \circled{$a_1$} & \dottedcircled{$a_2$} & $a_6$ \\
    $a_4$ & $a_1$ & \circled{$a_5$} & \dottedcircled{$a_6$} & $a_2$ & $a_3$ \\
    $a_5$ & \circled{$a_6$} & $a_2$ & \dottedcircled{$a_4$} & $a_1$ & $a_3$ \\
    $a_6$ & $a_1$ & $a_2$ & $a_3$ & \circled{$a_4$} & \dottedcircled{$a_5$} 
    \end{tabular}
\label{table:unsolvable}
\end{minipage} 
\end{table}

There are many practical applications of the {\sc sr} model. As the name suggests, it can model campus housing allocation where two students either share a room or a flat \citep{PPR08}. Furthermore, {\sc sr} can model pairwise kidney exchange markets \citep{MO12}, with centralised matching schemes existing, for example, in the US \citep{apd}, Netherlands \citep{Kleetal05ajt}, and UK \citep{dodtnhsbt}. {\sc sr} can also model peer-to-peer networks, such as file-sharing networks \citep{GMMR07}, and pair formation in chess tournaments \citep{KLM99}.

First proposed by \citet{gale_shapley}, the problem is, historically, a non-bipartite extension well-studied in its own right of the classical {\sc Stable Marriage} problem ({\sc sm}), where the set of agents can be partitioned into two sets such that every agent from one set has complete preferences over all agents from the other set, but does not accept to be matched with anyone from their own set. The authors showed that any {\sc sm} problem instance admits at least one stable matching, and that such a matching can be found in linear time.

In {\sc sr}, a problem instance $I$ can either be represented in the form of preference lists or as a labelled complete graph. Similarly, a matching $M$ can be represented either as a set of unordered pairs of agents or as a set of edges in the graph, respectively. These representations are equivalent. The {\sc sr} problem is well-studied \citep{gusfield89, matchup}, and \citet{irving_sr} presented an algorithm to find a stable matching or decide that none exists in linear time, which is referred to as \emph{Irving's algorithm}. Furthermore, there is a wide range of algorithms and structural results for special kinds of stable matchings, including algorithms for finding stable matchings that satisfy further optimality constraints on the structure of the matching. For example, the \emph{profile} is often used in fairness measures of a matching, which is a vector $(p_1 \dots p_n)$ where $p_i$ captures the number of agents that are matched to their $i$th choice, and $n$ is the number of agents. However, naturally, these optimality properties cannot be realised if the instance is unsolvable.

\citet{pittelirving94} presented some empirical evidence relating to solvability probability, denoted $P_n$, which is the probability that a uniformly random {\sc sr} instance admits at least one stable matching. Their results suggest that for instances with 100 to 2000 agents, $P_n$ appears to decline from around 64\% down to around 33\%. \citet{mertens05, mertens15random} extended the results through extensive simulations and conjectured that with an even number of agents $n$, $P_n \simeq e\sqrt{2\pi^{-1}} n^{-\frac{1}{4}}$. Thus, it is clear that as $n$ grows large, stable matchings are unlikely to be a good solution concept in practice, as they are unlikely to exist. In the past, many alternative solution concepts have been proposed, such as \emph{stable partitions} \citep{tan91_1}, \emph{maximum stable matchings} \citep{tan91_2}, \emph{almost-stable matchings} \citep{abraham06}, \emph{popular matchings} \citep{gupta_popsr_21}, and more, some of which leave many agents unassigned or are NP-hard to compute.

One of the fundamental {\sc sr}-related questions posed by \citet{gusfield89} asked about the existence of a succinct certificate for the unsolvability of an instance. The question was answered positively by \citet{tan91_1}, who generalised the notion of a stable matching to a new structure called a \emph{stable partition}. A stable partition is a cyclic permutation $\Pi$ of the agents, where every agent prefers their successor in $\Pi$ over their predecessor in $\Pi$, and no two agents strictly prefer each other over their predecessors (for example, $\Pi=(a_1 \; a_2\; a_3)(a_4\; a_5\; a_6)$ is the unique stable partition of Example \ref{table:unsolvable}, indicated in dotted and unbroken circles). Tan proved that cycles of odd length are invariant for all stable partitions of an instance and that an instance is unsolvable if and only if any stable partition contains a cycle of odd length. Over 20 years after their initial publication, \citet{matchup} described the work on stable partitions in the 1990s as a key landmark in the progress made on the {\sc sr} problem after 1989. 

Although proposed and studied as a succinct certificate for the unsolvability of an {\sc sr} instance, the notion of a stable partition was adopted as a solution concept in itself. Remaining in the tennis analogy, a stable partition can be interpreted as an assignment of the friends into half-hour sessions, or a weekly alternating schedule. The stability definition leads to a guaranteed solution in which every friend plays either for exactly one hour with some other friend or two half-hour sessions with two different friends, at most one agent does not play at all, and no two friends who would like to play more time with each other. In Example \ref{table:unsolvable}, agents $a_1,a_2$ and $a_3$ are all half-matched to each other, as are agents $a_4,a_5$ and $a_6$. Therefore, remaining in the tennis analogy, agent $a_1$ would play one half-hour match with each of $a_2$ and $a_3$, agent $a_2$ would play an additional half-hour match with $a_3$, and agents $a_4, a_5$ and $a_6$ would all play half-hour matches with each other too. Note that for any solvable instance, any stable partition can be transformed into a stable matching in linear time by decomposing any even-length cycles into transpositions, in which case the collection of transpositions corresponds to a stable matching. For unsolvable instances, it is guaranteed that some agents are in half-time partnerships due to the existence of odd-length preference cycles, which cannot be decomposed.

Some have highlighted the advantages of considering truly fractional stable matchings, whereas stable matchings are inherently integral, and stable partitions can be interpreted as stable half-integral matchings \citep{Chen2021}. While stable fractional matchings might be advantageous in applications such as time-sharing \citep{timesharing}, it is unlikely to be a suitable solution concept for other applications. For example, a tennis schedule for a set of friends where everyone plays many times for varying durations, which could be very short, is unlikely to be satisfactory.

Overall, stable partitions are an important structure and tool in the study of {\sc sr} and in the field of stable matchings. They can be used to show properties about instances, algorithms, and structures and pose an efficient solution concept to solvable and unsolvable instances at the intersection of stable integral and fractional matchings. Surprisingly, though, stable partitions themselves have not received nearly as much attention as stable matchings. This motivates our study of stable partition structures in more detail, specifically the investigation of different cycle types, the algorithmics of stable cycle and partition enumeration, how profile-optimality in the case of stable matching and maximum utility in the case of stable fractional matching adapts to stable partitions, and which of these variants are efficiently solvable or approximable.

\subsection{Our Contributions} 
\label{section:contributions}

In this paper, we prove new structural results for stable partitions, allowing us to construct algorithms to enumerate stable partitions and to compute or approximate various types of optimal stable partitions.

Specifically, we characterise the structure of \emph{reduced stable partitions} (i.e., stable partitions with no cycles of even length longer than 2) in a given {\sc sr} instance $I$ by proving a bijective correspondence between the set of reduced stable partitions and the set of stable matchings of a solvable sub-instance of $I$. With this, we show how the set of all reduced stable partitions can be enumerated efficiently. 

We build on these results by considering all stable partitions, not necessarily reduced. We show that any non-reduced stable partition can be constructed from two reduced stable partitions and leverage our previous result to enumerate all stable partitions. However, a deeper structural investigation yields the non-trivial result that any predecessor-successor pair can be part of at most one cycle of length longer than 2. This proves that the union of all stable partitions of any instance with $n$ agents contains at most $O(n^2)$ cycles, and we show how these can be enumerated in $O(n^4)$ time. We then use these structural insights to improve on the previous enumeration algorithm for all stable partitions, denoted $P(I)$, and present a more efficient algorithm that runs in $O(\vert P(I)\vert n^3 + n^4)$ time.

Finally, we adapt six natural fairness measures and optimality criteria from stable matchings to stable partitions. We prove that finding a stable partition that minimises the maximum predecessor rank (called \emph{minimum-regret stable partition}) is the only tractable problem variant of these, while the remaining five are NP-hard to compute. Furthermore, although two of the NP-hard optimisation problem variants are polynomial-time approximable within a factor of 2, one of them is not approximable within any constant factor. 

We put these results into perspective through experimental results collected by applying these novel algorithms to instances with preferences sampled uniformly at random. Altogether, the results outlined above are consistent with those previously known about stable matchings, strengthening the tight correspondence between stable matchings and stable partitions observed from the beginning of this work. Overall, our results shed more light on the complexity-theoretic and structural properties of stable partitions.

\subsection{Related Work}
\label{section:relatedwork}

\paragraph{Stable Matchings} 
\label{sec:rwsm}
\citet{irving_sr_structure, gusfield_sr_structure} published an extensive collection of early structural results, leading to a book discussing the structure and algorithms of the {\sc sm} and {\sc sr} problems by \citet{gusfield89}. They proved, for example, that the stable matchings of a given instance form a semi-lattice structure, that the same set of agents are unassigned in every stable matching of the same instance and that, of a given {\sc sr} instance $I$ with $n$ agents, all stable pairs (union of all pairs part of some stable matching) and all stable matchings, denoted $S(I)$, can be found in $O(n^3\log n)$ and $O(\vert S(I)\vert n^2 + n^3 \log n)$ time, respectively. Later, \citet{federegalapprox} used a {\sc 2-sat} reduction to obtain improved time complexities for the last two problems, solving them in $O(n^3)$ and $O(\vert S(I)\vert n + n^2)$ time, respectively.

With regards to the complexity of profile-optimal stable matching problems, \emph{minimum-regret} is one of the few efficiently solvable ones, with a linear-time algorithm presented by \citet{gusfield89}. On the other hand, \citet{federegal} proved that finding an egalitarian (minimum sum of ranks of agents assigned) stable matching for {\sc sr} is NP-hard. Later, \citet{federegalapprox}, and \citet{gusfieldegalapprox} gave 2-approximation algorithms for the problem using a reduction to {\sc 2-sat}. The hardness result is refined by \citet{csehshort} who established NP-hardness even when preference lists are of length at most 3. On the other hand, \citet{chen17} showed that the problem is fixed-parameter tractable with respect to egalitarian cost. \citet{CooperPhD} showed that for an instance with regret $r$, the problems of finding \emph{rank-maximal} (lexicographically maximal profile), \emph{generous} (lexicographically minimal inverse profile), \emph{first-choice maximal} (maximum number of first choices achieved), and  \emph{regret-minimal} (minimum number of agents assigned to $r$) stable matchings, if they exist, are all NP-hard. \citet{simola2021profilebased} extended these results to short preference lists and presented some approximability results.

As a natural way to deal with unsolvable instances, \citet{abraham06} introduced the problem of finding \emph{almost-stable} matchings, which are matchings with a minimum number of blocking pairs. They proved that the problem is NP-hard and not approximable within $n^{\frac{1}{2}-\varepsilon}$ for any $\varepsilon>0$, unless P$=$NP. Later, \citet{biro12} extended the study of this problem to bounded preference lists, showing that the problem remains inapproximable within some constant factor even for preference lists of lengths at most 3.

\paragraph{Stable Partitions}
In his original paper, \citet{tan91_1} provided a linear-time algorithm similar to Irving's algorithm to compute a stable partition, referred to as \emph{Tan's algorithm}, and showed that every {\sc sr} instance admits at least one stable partition. An alternative proof for the existence of stable half-matchings in more general settings was provided by \citet{aharonifleiner03}, and extended by \citet{birofleiner15} using Scarf's Lemma. \citet{tanhsueh} considered the online version of the problem of finding a stable partition, in which a new agent $a_i$ arrives and the preference lists are modified. They showed how to update, in linear time, an existing stable partition, for the instance prior to the arrival of $a_i$, to a new stable partition for the instance after $a_i$ arrives.  The entire process, in which a stable partition is built from scratch for all $n$ arriving agents, takes $O(n^3)$ time and is known as the Tan-Hsueh algorithm.

\citet{pittel93instance} derived a range of probabilistic results about the original algorithm by Tan. Specifically, he showed, for example, that every stable partition is likely to be almost a stable matching, in the sense that at most $O(\sqrt{(n\log n)})$ members are likely to be involved in the cycles of odd length three or more. He also showed that the expected number of stable partitions is $O(\sqrt{n})$. This work was extended by \citet{pittel19}, who showed that the expected total number of reduced stable partitions grows proportionally to $O(n^{\frac{1}{4}})$. Pittel also showed that the expected total number of cycles of odd length grows proportionally to $O(n^{\frac{1}{4}}\log n)$ and that, with high probability, the total length of all its cycles of odd length, denoted $n_{odd}$, is below $\sqrt{n}\log n$.

With regards to experimental results related to components of stable partitions, \citet{mertens05} analysed $n_{odd}$ for unsolvable instances sampled from complete uniform random graphs and conjectured that the expected total number is $\Theta(\sqrt{n (\log n)^{-1}})$. Furthermore, \citet{mertens15small} presented the exact probabilities for specific component types for small instances. 

Stable partitions are also used to show various other results with regard to the {\sc sr} problem. For example, \citet{biro08} proved some useful properties of the algorithm by \citet{tanhsueh}. A key result is that in the algorithm, the agent that arrives last gets their best possible partner in any stable matching. Similarly, in their study of almost-stable matchings, \citet{abraham06} used stable partitions to prove bounds on the number of blocking pairs. Finally, stable partitions are also used in exact and approximation algorithms for finding almost-stable matchings by \citet{biro12}.

\paragraph{Similar Problems and Extensions} 
Over time, many variations of {\sc sr} have been studied, such as {\sc srt} (where ties are allowed in the preferences) \citep{irving_srt, ronn_srt, ScottPhD, kunysz_srt_strongstable}, {\sc sri} (where incomplete preferences are permitted) \citep{csehshort}, or both, denoted {\sc srti} \citep{irving_srt}. Another extension of {\sc sr} is the problem of finding stable fractional matchings. One advantage of fractional matchings is the natural integration of cardinal utilities rather than purely ordinal preferences, as highlighted by \citet{Anshelevich2009}. \citet{Caragiannis2019} were the first to study a combined model of stable fractional matchings with cardinal preferences in the {\sc sm} context from a computational complexity point of view. The authors established the NP-hardness of computing an optimal solution in terms of a stable fractional matching with the highest overall utility achieved. As a special case of {\sc sr}, these hardness results carry over to the {\sc sri} problem. Recently, \citet{Chen2021} expanded on these results and completed the complexity classification of optimising for overall utility or the number of fully matched agents in both the {\sc sm} and {\sc sr} settings with and without ties. They established that a stable fractional matching always exists, but that both optimisation problems are NP-hard when considering cardinal utilities. Furthermore, they showed that stable partitions are a 2-approximation of the optimisation problem of finding stable fractional matchings with the maximum number of agents fully matched in the presence of ties. We note that these papers discuss different models of stability in the fractional setting.

Building on theory for bipartite problem restrictions by \citet{crawford81} and \citet{kelso82}, \citet{Fleiner2010} showed how, under some assumptions, the strict linearly-ordered preference lists can be replaced with choice functions, which are a generalisation of ordinal rankings and can capture more complex preference behaviours, while keeping the problem polynomial-time solvable. \citet{Irving2007} introduced a many-to-many extension of {\sc sr} called {\sc Stable Fixtures} and provided a linear-time algorithm to find a stable matching or report that none exists. This was later extended by \citet{biro18} to a model with payments. A similar many-to-many extension to {\sc Stable Fixtures}, where, additional, parallel edges are allowed, i.e., agents can be matched to each other multiple times, was initially studied by \citet{cechlarova05} and later extended in \citep{bmatchingrotations}. \citet{deanmunshi10} investigated the even more general {\sc Stable Allocation} problem and used transformations between non-bipartite and bipartite instances to derive efficient algorithms.

As an alternative to stability, a matching $M$ is \emph{popular} if there does not exist some matching $M'$ which is preferred by more agents than those that prefer $M$. \citet{gupta_popsr_21, faenza_pop} simultaneously proved that the problem of deciding whether a {\sc Stable Roommates} instance admits a popular matching is NP-complete. \citet{schlotter22} extended this result to maximise the utility of a matching that is popular but admits only a few blocking pairs, proving that finding a popular matching with at most one blocking pair is already NP-hard.

\subsection{Structure of the Paper} 
\label{section:paperstructure}

In Section \ref{section:formaldefinitions}, we present a mixture of existing and new formal definitions and known results relevant to this paper. In Section \ref{section:cyclesandpartitions}, we characterise and exploit the structure of reduced stable partitions, longer cycles, and all stable partitions. In Section \ref{section:complexity}, we adapt common profile-based optimality criteria from stable matchings to stable partitions and investigate the complexity and approximability of these problems. Lastly, in Section \ref{section:experiments}, we present experimental results generated with the novel algorithms described here. We finish with a discussion of the results and some related open problems in Section \ref{section:discussion}.

\section{Formal Definitions and Known Results}
\label{section:formaldefinitions}

In Section \ref{section:problemsolutiondefinitions}, we will outline some crucial concepts and their notation with regard to stable matchings, stable partitions, and their different components. Following, Section \ref{section:fairnessmeasures} will introduce some natural fairness and optimality measures previously studied in the literature for stable matchings and adapt them for our study of stable partitions.

\subsection{Stable Matchings, Partitions, and Cycles}
\label{section:problemsolutiondefinitions}

We begin this section by defining {\sc sr} instances formally as follows.

\begin{definition}[{\sc sr} Instance]
    Let $I=(A,\succ)$ be an {\sc sr} \emph{instance} where $A=\{ a_1, a_2, \dots, a_n \}$, also denoted $A(I)$, is a set of $n\in2\mathbb{N}$ agents and every agent $a_i\in A$ has a strict preference ranking (a linear \emph{preference relation}) $\succ_i$ over all other agents $a_j\in A\setminus\{a_i\}$. For stable partitions, we will assume that $\succ$ is extended such that every agent ranks themselves last.
\end{definition}

Note that in this work, we assume an even number of agents. However, it has been shown \citep{gusfield89} that most concepts and results transfer over to the case where the number of agents is odd, and we accept an unmatched agent (who cannot be in a blocking pair). Furthermore, there will be cases where we transform an {\sc sr} instance with \emph{complete preferences} (i.e. a ranking $\succ$ in which all agents rank all other agents but themselves) into an instance with \emph{incomplete preferences} (often denoted by {\sc sri}), but all existing techniques that we use still apply in the {\sc sri} case. 

A solution to an {\sc sr} instance is now defined.

\begin{definition}[Stable Matchings]
    Let $I=(A,\succ)$ be an {\sc sr} instance. $M$ is a \emph{matching} of $I$ if it is an assignment of some agents in $A$ into pairs such that no agent is contained in more than one pair of $M$. A \emph{blocking pair} of a matching $M$ is a pair of two agents $a_i, a_j\in A$ such that either $a_j$ is unassigned in $M$ or $a_i \succ_j M(a_j)$, and either $a_i$ is unassigned or $a_j \succ_i M(a_i)$, where $M(a_i)$ is the partner of $a_i$ in $M$. If $M$ does not admit any blocking pair, then it is called \emph{stable}. A stable matching is \emph{complete} if it contains all $\vert A\vert$ agents.
\end{definition}

Clearly, any stable matching of an instance with complete preferences and an even number of agents will be complete, as there can only be an even number of unmatched agents, and any two unmatched agents would rather be matched to each other than to be unmatched. 

\begin{definition}[Solvability and Sub-Instances]
    Let $I=(A,\succ)$ be an {\sc sr} instance. $I$ is \emph{solvable} if it admits at least one stable matching, otherwise it is \emph{unsolvable}. $I'=(A', \succ')$ is a \emph{sub-instance} or a \emph{restriction} of $I$ if $A'\subseteq A$ and $a_i \succ'_k a_j$ for all agents $a_i, a_j, a_k\in A'$ where $a_i \succ_k a_j$.
\end{definition}

As previously mentioned, stable partitions $\Pi$ are permutations on the subset of agents. Therefore, they can be written in cyclic notation where every agent is part of exactly one cycle. A formal definition of what it means to be blocking follows. 

\begin{definition}[Permutation-Blocking Pairs]
    Let $I=(A,\succ)$ be an {\sc sr} instance and let $\Pi$ be a permutation of $A$. Then two agents $a_i, a_j$ \emph{block} (form a \emph{blocking pair} of) $\Pi$ if $a_j\succ_i \Pi^{-1}(a_i)$ and $a_i\succ_j \Pi^{-1}(a_j)$. We can also say that (without loss of generality) $a_i$ \emph{blocks} $\Pi$ \emph{with} $a_j$. 
\end{definition}

Note that the definition of blocking agents with respect to a permutation is very similar to the corresponding definition in the context of matchings, and it should be clear from the context which definition applies.

\begin{definition}[Stable Partition]
    Let $I=(A,\succ)$ be an {\sc sr} instance. Then a permutation $\Pi$ of $A$ is referred to as a \emph{stable partition} if  
    \begin{enumerate}
        \item[(T1)] $\forall a_i \in A$ we have $\Pi(a_i) \succeq_i \Pi^{-1}(a_i)$, and
        \item[(T2)] $\nexists.a_i, a_j \in A, \; a_i\neq a_j,$ such that $a_i,a_j$ block $\Pi$,
    \end{enumerate}
    where $\Pi(a_i) \succeq_i \Pi^{-1}(a_i)$ means that either $a_i$'s successor in $\Pi$ is equal to its predecessor, or the successor has a better rank than the predecessor in the preference list of $a_i$.
\end{definition}

Revisiting the definitions of blocking pairs in the contexts of matchings and permutations, we observe that the similarity is natural because each stable partition $\Pi$ corresponds to a stable half-matching in the sense that each successor-predecessor pair in $\Pi$ is assigned a half-integral match (such that two agents are fully matched precisely if they are in a transposition). A blocking pair with respect to a stable partition corresponds precisely to two agents that would prefer to be assigned more units to each other in the corresponding half-matching.

In the study of stable partitions, it is important to differentiate between different kinds of cycles that may or may not be part of a stable partition.

\begin{definition}[Cycles]
    Let $I$ be an {\sc sr} instance and let $C=(a_{i_1} \; a_{i_2} \;\dots\; a_{i_k})$ be an ordered collection of one or more agents in $I$. Then $C$ is a \emph{cycle} and we can \emph{apply} cycle $C$ to some agent $a_{i_j}\in C$ to get its successor in $C$, denoted by $C(a_{i_j})$. Similarly, we can apply the inverse $C^{-1}=(a_{i_1} \; a_{i_k} \; \dots \; a_{i_2})$ to $a_{i_j}$ to get its predecessor in $C$, denoted by $C^{-1}(a_{i_j})$. The same holds for sets of disjoint cycles. 
\end{definition}

We will implicitly assume two disjoint cycles $C_1, C_2$ can be added under concatenation $C_1$  $\doubleplus$$C_2$, which will simply be denoted by $C_1C_2$ when obvious. A cycle $C_1$ can be removed from a collection $C$ of cycles using rules of set difference, simply denoted by $C\setminus C_1$.

\begin{definition}[Stable Cycles]
\label{def:stablecycles}
    Let $I$ be an {\sc sr} instance and let $C$ be a cycle. If there exists some stable partition $\Pi$ of $I$ such that $C\in \Pi$, then $C$ is a \emph{stable cycle} of $I$. A stable cycle of odd length is called an \emph{odd cycle}, of length 2 a \emph{transposition}, and of even length longer than 2 an \emph{even cycle} of $I$. Clearly, for every stable partition $\Pi$ of $I$, $\Pi=\mathcal{T} \mathcal{E} \mathcal{O}_I$, where $\mathcal T$ denotes the transpositions, $\mathcal E$ the even cycles, and $\mathcal{O}_I$ the odd cycles of $\Pi$. Let $A(C)$ denote the set of agents contained in a single cycle or a collection $C$ of cycles. Unless specified otherwise, let $n_1=\vert A(\mathcal T$  $\doubleplus$$\mathcal E)\vert$ and $n_2=\vert A(\mathcal{O}_I)\vert$. 
\end{definition}

We choose to denote the odd cycles $\mathcal{O}_I$ with a subscript $I$ to emphasise that these are always invariant between different stable partitions and are an inherent property of the instance $I$ (this structural result will be stated in Theorem \ref{thm:tan} below). On the other hand, the transpositions $\mathcal{T}$ and even cycles $\mathcal{E}$ can vary between stable partitions.

The following definition makes a further distinction between stable partitions with and without even cycles.

\begin{definition}[Reduced Stable Partition]
    A cycle is called \emph{reduced} if its length is either 2 or odd, otherwise it is called \emph{non-reduced}. Similarly, a stable partition $\Pi$ is \emph{reduced} if it consists only of reduced cycles, and \emph{non-reduced} if not.
\end{definition}

We will also need a definition for sub-sequences of cycles.

\begin{definition}[Partial Cycles]
    Let $I$ be an {\sc sr} instance and let $C=(a_{i_1} \; a_{i_2}  \dots a_{i_k} \dots)$ be a cycle where only some agents $a_{i_1} \dots a_{i_k}$ are known to occur consecutively in the cycle. Then we call $C$ a \emph{partial cycle}. If there exist agents $a_{i_{k+1}}\dots a_{i_r}$ such that $(a_{i_1} \; a_{i_2} \dots a_{i_k} \; a_{i_{k+1}} \dots a_{i_r})$ is a stable cycle, then $C$ is a \emph{partial stable cycle} and the sequence of agents $a_{i_{k+1}} \dots a_{i_r}$ is a \emph{completion} of $C$.
\end{definition}

The following fundamental results regarding stable partitions will be assumed throughout this work.

\begin{theorem}[\citep{tan91_1, tan91_2}]
\label{thm:tan}
    The following properties hold for any {\sc sr} instance $I$.
    \begin{itemize}
        \item $I$ admits at least one stable partition, and any non-reduced stable partition can be transformed into a reduced one by breaking down its longer even-length cycles into collections of transpositions.
        \item Any two stable partitions $\Pi_a, \Pi_b$ of $I$ contain exactly the same odd cycles, not only involving the same sets of agents, but also ordering the agents in the same way within the odd cycles.
        \item $I$ admits a complete stable matching if and only if no stable partition of $I$ contains an odd cycle.
    \end{itemize}
\end{theorem}

Note that if an {\sc sr} instance $I$ admits a stable matching $M=\{\{a_{i_1}, a_{i_2}\}, \dots, \{a_{i_{2k-1}}, a_{i_{2k}}\}\}$ consisting of $k$ pairs, then it can be denoted interchangeably with its induced collection of transpositions $\Pi=(a_{i_1}$ $a_{i_2})\dots (a_{i_{2k-1}}$ $a_{i_{2k}})$, which is a stable partition of $I$.

\subsection{Optimality and Fairness of Matchings and Partitions}
\label{section:fairnessmeasures}

As previously noted, there are different measures of fairness and optimality for stable matchings considered in the literature. The main measure is the \emph{profile} $p(M)$ of a matching $M$, a vector $(p_1 \dots p_{n-1})$, where $p_i$ counts the number of agents assigned to their $i$th choice in $M$. We define a similar measure for stable partitions, accounting for all half-assignments in the associated half-matching.

\begin{definition}[Profile]
    Let $I$ be an {\sc sr} instance with $n$ agents and $\Pi$ a stable partition of $I$. An agent $a_i$ has \emph{rank} $i$ (where $i\geq 1$) for an agent $a_j$ if it appears at position $i-1$ (where position indices count from 0) in $a_j$'s preference list. The \emph{successor profile}, denoted $p^s(\Pi) = (p_1^s \; \dots \; p_n^s)$, captures the number of agents $p_i^s$ whose successor in $\Pi$ has rank $i$. Analogously, let $p^p(\Pi)$ denote the \emph{predecessor profile} of $\Pi$. Finally, let $p(\Pi)= (p_1 \; \dots \; p_n)$ denote the combined \emph{profile} of $\Pi$, where $p_i=p_i^s+p_i^p$. Similarly, for a stable cycle $C$, let $p^s(C), p^p(C),$ and $p(C)$ denote the profiles containing the number and positions of predecessors, successors, or both achieved in $C$.
\end{definition}

For any stable partition $\Pi$ of $n$ agents, the sum of all entries in $p^s(\Pi)$ and $p^p(\Pi)$ add up to $n$ each, such that those in $p(\Pi)$ add up to $2n$, one for each half-assignment and counting a full assignment as two half-assignments. 

One simple measure of optimality for a stable matching $M$ of a solvable instance $I$ is its \emph{regret}, which is the rank of the worst assigned agent in $M$. Similarly, the regret of $I$ is the minimum regret over all its stable matchings. We will consider the same for stable partitions.

\begin{definition}[Regret]
    Let $\Pi$ be a stable partition and let $p(\Pi)$ be its profile. Then the \emph{regret} of $\Pi$, denoted $r(\Pi)$, is the rank of the worst-ranked agent assigned in $\Pi$, which is the index of the last positive entry in $p(\Pi)$. Similarly, for $C$ some cycle of agents, let $r(C)$ be the rank of the worst-ranked agent assigned in $C$.
\end{definition} 

Finally, another important optimality measure of a stable matching $M$ is its \emph{cost}, which is a weighted sum of its profile. Again, the same can be defined for stable partitions. 

\begin{definition}[Cost]
    Let $\Pi$ be a stable partition and let $p(\Pi)$ be its profile. Then the \emph{cost} of $\Pi$, denoted $c(\Pi)$, is the sum of ranks achieved in $\Pi$, averaged over the predecessor and successor assignments. Specifically, we define $c(\Pi) = \frac{1}{2} \sum_{1\leq i < n} p_i(\Pi)\cdot i$.
\end{definition}

We include the factor of half for consistency with stable matchings, as the profile entries sum up to $2n$. This constant factor does not affect our optimisation problems. For the same reason, we also do not sum up to $n$, as an agent assigned to someone of rank $n$ is assigned to itself, which can never happen in stable matchings. As any two self-assigned agents would block in an instance with complete preference lists, there can be at most one self-assigned agent in a stable partition, and this is an odd cycle and thus invariant, so this again does not affect our optimisation problems.

Some associated problem variants for stable matchings are the following. 

\begin{definition}
    Let $I$ be a solvable {\sc sr} instance. A \emph{minimum-regret} stable matching is a stable matching of $I$ with regret $r$ such that no other stable matching has lower regret. A \emph{regret-minimal} stable matching is a minimum-regret stable matching such that no other minimum-regret stable matching satisfies fewer choices of rank $r$. A \emph{first-choice maximal} stable matching is a stable matching of $I$ such that no other stable matching satisfies more first choices. Finally, a (egalitarian) \emph{cost-optimal} stable matching of $I$ is a stable matching such that no other stable matching has a lower cost.
\end{definition}

With our definitions of profile, regret, and cost, we can define analogous optimality measures for stable partitions (and thus stable half-matchings). For some, we need an order on profiles as defined below.

\begin{definition}
    Let $p=(p_1 \dots p_n),p'=(p_1' \dots p_n')$ be two profile vectors. Then $p=p'$ if $p_i=p_i'$ for all $1\leq i\leq n$. If $p\neq p'$, let $k$ be the first position in which they differ. We define $p\succ p'$ if $p_k>p_k'$. Furthermore, we define $p\succeq p'$ if either $p\succ p'$ or $p=p'$. Finally, we will call $p^{rev}=(p_n \dots p_1)$ the \emph{reverse profile} of $p$.
\end{definition}

We define the optimality criteria for stable partitions as follows. Recall from Section \ref{sec:rwsm} that all of these optimality notions below have previously been investigated in the context of stable matchings.

\begin{definition}
    Let $\Pi$ be a stable partition of some {\sc sr} instance $I$ with $n$ agents and regret $r$, and let $p(\Pi)$ be its profile. Then $\Pi$ is a 
    \begin{itemize}
        \item \emph{minimum-regret} stable partition of $I$ if, for all stable partitions $\Pi'$ of $I$, we have $r(\Pi)\leq r(\Pi')$,
        \item \emph{first-choice maximal} stable partition of $I$ if, for all stable partitions $\Pi'$ of $I$ with profile $p(\Pi')$, we have $p_1(\Pi) \geq p_1(\Pi')$,
        \item \emph{rank-maximal} stable partition of $I$ if, for all stable partitions $\Pi'$ of $I$ with profile $p(\Pi')$, we have $p(\Pi) \succeq p(\Pi')$,
        \item \emph{regret-minimal} stable partition of $I$ if $\Pi$ is a minimum-regret stable partition and, for all minimum regret stable partitions $\Pi'$ of $I$ with profile $p(\Pi')$, we have $p_r(\Pi) \leq p_r(\Pi')$,
        \item \emph{generous} stable partition of $I$ if, for all stable partitions $\Pi'$ of $I$ with profile $p(\Pi')$, we have $p^{rev}(\Pi')\succeq p^{rev}(\Pi)$, and
        \item \emph{egalitarian} stable partition of $I$ if, for all stable partitions $\Pi'$ of $I$, we have $c(\Pi)\leq c(\Pi')$.
    \end{itemize}
\end{definition}

\section{Structure of Stable Cycles, Transpositions, and Partitions}
\label{section:cyclesandpartitions}

First, we will formalise and characterise the structure of reduced stable partitions in Section \ref{section:reducedpartitionsstructure}. Then, naturally, we will extend the study to all stable partitions. This requires the study of stable cycles. After showing various results about the structure of even cycles in Section \ref{section:longercycles}, we focus on algorithmic questions related to the enumeration of all reduced and non-reduced stable cycles in Section \ref{section:allcycles}, finishing this section with an enumeration algorithm for all stable partitions in Section \ref{section:enumeratingpartitions} which ties together most of our results.

\subsection{Structure of Reduced Stable Partitions}
\label{section:reducedpartitionsstructure}

Let $I=(A,\succ)$ be an instance of {\sc sr} and let $\Pi_i$ be any reduced stable partition of $I$. We know that its odd cycles $\mathcal O_I$ are invariant under all stable partitions of $I$, while the transpositions might not be. Let $I_E=(A', \succ')$, where $A'=A \setminus A(\mathcal O_I)$ and $\succ'$ is the restriction of $\succ$ to the agents in $A'$, be the instance $I$ restricted to agents in transpositions. 

\begin{lemma}
\label{lemma:reducedstable}
    Let $\Pi_i$ be any reduced stable partition of an {\sc sr} instance $I$. Then $\Pi_i$ is the union $M_i\mathcal{O}_I$ of its odd cycles $\mathcal{O}_I$ and a complete stable matching $M_i$ of the sub-instance $I_E$ of $I$ constructed above as a cyclic permutation of transpositions.
\end{lemma}
\begin{proof}  
    By definition of reduced stable partitions, $\Pi_i$ is the union $M_i\mathcal O_I$ of transpositions $M_i$ and its invariant odd cycles $\mathcal O_I$. Note first that if $M_i=\varnothing$, then $A(\mathcal{O}_I)=A$, so $A'=\varnothing$ and $I_E=\varnothing$. Therefore, $M_i$ is trivially complete and stable in $I_E$.

    Otherwise, suppose that $M_i\neq \varnothing$ with $\vert A'\vert \in 2\mathbb{N}$, where $M_i$ is not a stable matching in $I_E$. Note first that $M_i$ is a matching in $I_E$ as the transpositions are disjoint due to $\Pi_i$ being a partition, and only include agents in $A'$ by construction. Then either some partners $a_r, a_s$ in $M_i$ would not find each other acceptable anymore in $\succ'$, or there would be at least one blocking pair consisting of two agents $a_r, a_s$ in $A'$ that prefer each other over their partners in $M_i$ in the preference relation $\succ'$.

    Consider first the case where some $\{a_r, a_s\} \in M_i$ would not find each other acceptable anymore in $\succ'$. This cannot happen, as $A(M_i)\cap A(\mathcal{O}_I)=\varnothing$, so the pair would not be removed in the construction of $\succ'$, where $\succ'$ is restricted to agents not in $\mathcal{O}_I$.

    Furthermore, $M_i$ cannot admit a blocking pair in $I_E$. Suppose that $\{a_r, a_s\} \subseteq A'$ was a blocking pair in $I_E$, then $a_r\succ'_s M_i(a_s)$ and $a_s\succ'_r M_i(a_r)$. However, because all agents in $A'$ are contained in transpositions in $\Pi_i$, we know that $M_i(a_r)=\Pi_i(a_r)=\Pi_i^{-1}(a_r)$ and $M_i(a_s)=\Pi_i(a_s)=\Pi_i^{-1}(a_s)$. Because $\succ'$ is contained in $\succ$, this would let us conclude that $a_r\succ_s \Pi_i^{-1}(a_s)$ and $a_s\succ_r \Pi_i^{-1}(a_r)$, a contradiction of stability condition T2 of $\Pi_i$.

    Finally, every agent has complete preference lists in $I_E$, so if $M_i$ is stable, it must be a complete matching. Otherwise, there would be two unmatched agents in $A'$ that would rather be matched to each other than be unmatched, contradicting the stability of $M_i$. 
\end{proof}

On the other hand, we are not guaranteed that there is a bijection between the set $P(I)$ of all reduced stable partitions of $I$ and the set $S(I_E)$ of all stable matchings of $I_E$. Consider the instance $I$ in Example \ref{table:bijectionfail}. It admits the reduced stable partition $\Pi_0=(a_1$ $a_2)(a_3$ $a_4)(a_5$ $a_6$ $a_7)(a_8$ $a_9$ $a_{10})$. Now clearly $\mathcal O_I = (a_5$ $a_6$ $a_7)(a_8$ $a_9$ $a_{10})$ and $I_E=(\{a_1, a_2, a_3, a_4\}, \succ')$. However, $I_E$ admits the stable matchings $S(I_E) = \{ M_1, M_2 \}$ where $M_1 = \{ \{a_1, a_2\}, \{a_3, a_4\} \}$ and $M_2= \{ \{a_1,a_3\}, \{a_2,a_4\} \}$. While $M_1 \mathcal O_I$ is the original stable partition $\Pi_0$, $M_2\mathcal O_I$ does not correspond to a stable partition of $I$ as agents $\{a_1, a_5\}$ would violate condition T2. Specifically, $a_5 \succ_1 a_3$ and $a_1 \succ_5 a_7$. In Example \ref{table:bijectionfail}, this is indicated by $M_2$ in dotted circles and the violating pair in unbroken circles. 

\begin{table}[!htb]
\centering
\caption{{\sc sr} instance with stable partition $\Pi_0$ as \textbf{successors} and \underline{predecessors}}
    \begin{tabular}{ c | c c c c }
    $a_1$ & \underline{$\mathbf{a_2}$} & \circled{$a_5$} & \dottedcircled{$a_3$} & \dots \\
    $a_2$ & \dottedcircled{$a_4$} & \underline{$\mathbf{a_1}$} & \dots &  \\
    $a_3$ & \dottedcircled{$a_1$} & \underline{$\mathbf{a_4}$} & \dots &  \\
    $a_4$ & \underline{$\mathbf{a_3}$} & \dottedcircled{$a_2$} & \dots &  \\
    $a_5$ & \circled{$a_1$} & $\mathbf{a_6}$ & \underline{$a_7$} & \dots \\
    $a_6$ & $\mathbf{a_7}$ & \underline{$a_5$} & \dots &  \\
    $a_7$ & $\mathbf{a_5}$ & \underline{$a_6$} & \dots &  \\
    $a_8$ & $\mathbf{a_9}$ & \underline{$a_{10}$} & \dots &  \\
    $a_9$ & $\mathbf{a_{10}}$ & \underline{$a_8$} & \dots &  \\
    $a_{10}$ & $\mathbf{a_8}$ & \underline{$a_9$} & \dots &  \\
    \end{tabular}
\label{table:bijectionfail}
\end{table}

To deal with this case where enumerating the stable matchings of $I_E$ is not sufficient to derive the set $P(I)$, we will give a procedure to reduce $I_E$ further. Specifically, all pairs of agents $\{a_i, a_j\} \subseteq A'$ need to be deleted from the preference relation $\succ'$ whenever there exists an agent $a_r\in A(\mathcal O_I)$ such that $a_i \succ_r \mathcal O_I^{-1}(a_r)$ and $a_r \succ_i a_j$. We refer to this as making $\{a_i, a_j\}$ \emph{mutually unacceptable}, knowing that all pairs of agents are \emph{acceptable}, i.e., present in each other's preference list, at the outset due to our assumption of complete preference lists in the input instance and in the reduced instance $I_E$ (relative to a potentially smaller set of agents). This is equivalent to truncating the preference list of $a_i$ as follows. Pick the best-ranked agent $a_r\in A(\mathcal O_I)$ in $\succ_i$, and remove all agents $a_s$ that $a_i$ ranks lower than $a_r$ in $\succ_i'$. Simultaneously, remove $a_i$ from $\succ_s'$ for each such $a_s$. The deletion operations do not reduce the set of agents, just the preference relation, so we let the resulting sub-instance of $I_E$ be denoted by $I_T=(A', \succ'')$ (Algorithm \ref{alg:itconstruction} captures all steps to compute $I_T$ from $I$). We assume that preference comparison and partner lookup can be done in constant time, such that the procedure takes at most linear time overall.

\begin{algorithm}[!htb]
%\floatname{algorithm}{Procedure}
\renewcommand{\algorithmicrequire}{\textbf{Input:}}
\renewcommand{\algorithmicensure}{\textbf{Output:}}

    \begin{algorithmic}[1]

    \Require{$I$ : an {\sc sr} instance}
    \Ensure{$I_T$ : a solvable sub-instance of $I$}

    \State $I_T \gets I$
    \State $M_0\mathcal{O}_I \gets$ {\sf ReducedStablePartition$(I)$}

    \For{$a_r$ in $A(\mathcal{O_I})$}
        \State {remove $a_r$ from $I_T$}
    \EndFor

    \For{$a_i$ in $A(M_0)$}
        \State {\tt delete} $\gets$ False
        \For{$a_j$ in preference list of $a_i$ in $I$ (in order or preference)}
            \If{$a_j$ in $A(\mathcal{O}_I)$ \textbf{and} $a_i\succ_j \mathcal{O}^{-1}(a_j)$}
                \State {\tt delete} $\gets$ True
            \EndIf
            \If{{\tt delete} \textbf{and} $\{a_i, a_j\}$ acceptable in $I_T$}
                \State make $\{a_i, a_j\}$ mutually unacceptable in $I_T$
            \EndIf
        \EndFor
    \EndFor

    \State\Return{$I_T$}

    \end{algorithmic}
    \caption{{\sf ConstructI}$_T(I)$, constructs $I_T$ from {\sc sr} instance $I$}
    \label{alg:itconstruction}
\end{algorithm}

Although these deletions might destroy some stable matchings admitted by $I_E$, we must ensure that $I_T$ still admits at least one.

\begin{lemma}
    \label{lemma:m0isstable}
    Let $\Pi_i=M_i\mathcal{O}_I$ be any reduced stable partition of an {\sc sr} instance $I$ and let $I_T$ be its sub-instance as constructed above. Then $M_i$ is a complete stable matching of $I_T$.
\end{lemma}
\begin{proof}
    By Lemma \ref{lemma:reducedstable}, $M_i$ is a complete stable matching of $I_E$. As $\succ''$ is contained in $\succ'$, $M_i$ cannot admit more blocking pairs in $I_T$ than in $I_E$. Thus, if $M_i$ were not a stable matching in $I_T$, then this would have to be not because $M_i$ is not stable but because $M_i$ is not a matching. This would only occur if two agents in $M_i$ were made mutually unacceptable in the construction of $I_T$ from $I_E$.

    Suppose that $\{a_i, a_j\} \in M_i$ find each other acceptable in $\succ'$, but do not find each other acceptable anymore in $\succ''$. Then their respective preference list entries would have been removed in the deletion operations. Without loss of generality, suppose this was because there exists some agent $a_r\in \mathcal{O}_I$ such that $a_i \succ_r \mathcal{O}_I^{-1}(a_r)$ and $a_r \succ_i a_j$. Then $\{a_i, a_j\}$ could not belong to $M_i$ as $a_i$ and $a_r$ would contradict stability condition T2 with respect to $\Pi_i$. Therefore, matched agents in $M_i$ still consider each other mutually acceptable in $\succ''$.

    Thus, $M_i$ is also a complete stable matching of $I_T$.
\end{proof}

We can now show the correspondence between the stable matchings of $I_T$ and the reduced stable partitions of $I$.

\begin{theorem}
\label{theorem:correspondence}
    $M_i$ is a complete stable matching of $I_T$ if and only if $M_i\mathcal{O}_I$ is a reduced stable partition of $I$, where $I_T$ is constructed as above and $\mathcal{O}_I$ are the invariant odd cycles of $I$.
\end{theorem}
\begin{proof}
    First, suppose that $M_i$ is a complete stable matching of $I_T$. We want to show that $M_i\mathcal{O}_I$ is a reduced stable partition of $I$. Suppose not, then either $M_i\mathcal{O}_I$ is not a partition, not reduced, or it violates one of the stability conditions T1 or T2. 
    
    $M_i\mathcal{O}_I$ would not be a partition of $I$ if either there was an agent $a_r\in A$ such that $a_r\notin A(M_i\mathcal{O}_I)$ or if there was an agent $a_r\in A(M_i\mathcal{O}_I)$ that is in multiple different cycles. The first case cannot happen because $A = A' \uplus A(\mathcal{O}_I)$ and $M_i$ is a complete stable matching by assumption. Hence, every agent $a_r\in A$ is either in $M_i$ or in $\mathcal{O}_I$ and so $a_r \in A(M_i\mathcal{O}_I)$. The second case cannot happen by construction. If $a_r \in\mathcal{O}_I$, then by construction of $\mathcal{O}_I$ and $I_E$, $a_r$ only appears once in $\mathcal{O}_I$ and $a_r\notin A'$, so $a_r\notin M_i$. If $a_r\in A'$, then it was not deleted during the construction of $I_E$ and so $a_r\notin \mathcal{O}_I$. Because $M_i$ is a matching, $a_r$ will only be matched to one other agent and thus appear only once in the induced collection of transpositions. This also confirms that the partition is reduced because all cycles are either of odd length or of length two.

    Condition T1 cannot be violated by $M_i\mathcal{O}_I$ as $\mathcal{O}_I$ was found in such a way that T1 holds, and because $M_i$ is a collection of transpositions, which means that for all agents $a_r \in A'$, $a_r$'s predecessor is also their successor. 
    
    Finally, suppose that $M_i\mathcal{O}_I$ violates T2 with a pair of agents $a_r,a_s \in A = A' \uplus A(\mathcal{O}_I)$, then there are three cases. First, suppose that $a_r, a_s\in A'$. This means that the agents are not partners in $M_i$ but would rather be matched to each other. This cannot happen due to the stability of $M_i$ in $I_T$. Now suppose that $a_r, a_s \in A(\mathcal{O}_I)$. This cannot happen, as $\mathcal{O}_I$ was chosen in such a way that it satisfies T2. Lastly, suppose that, without loss of generality, $a_r\in A'$ and $a_s\in A(\mathcal{O}_I)$. Then $a_s \succ_r ( M_i\mathcal{O}_I)^{-1}(a_r)=M_i(a_r)$ and $a_r \succ_s (M_i\mathcal{O}_I)^{-1}(a_s) = \mathcal{O}_I^{-1}(a_s)$. However, in this case, the pair $\{a_r, M_i(a_r)\}$ would have been deleted from $\succ'$ during the construction of $I_T$ from $I_E$ and thus $\{a_r, M_i(a_r)\}$ could not be matched in $M_i$, a contradiction. 
    
    Now for the reverse implication, suppose that $M_i\mathcal{O}_I$ is a reduced stable partition of $I$, where $M_i$ is the collection of transpositions, if any, and $\mathcal{O}_I$ is the collection of odd cycles. By Lemma \ref{lemma:m0isstable}, $M_i$ is a complete stable matching of $I_T$.
\end{proof}

\begin{corollary}
    The number of reduced stable partitions that an {\sc sr} instance $I$ with $n\in 2\mathbb{N}$ agents admits is equal to the number of stable matchings that its sub-instance $I_T$ with $n - \vert A(\mathcal{O}_I) \vert$ agents admits.
\end{corollary}

With these results, we now have a procedure to enumerate all reduced stable partitions.

\begin{theorem}
\label{theorem:findallrp}
    We can enumerate all reduced stable partitions of an {\sc sr} instance $I=(A,\succ)$ in $O(\vert S(I_T) \vert n + n^2)$ time, where $I_T$ is the sub-instance of $I$ as constructed above, $S(I_T)$ is the set of all stable matchings of $I_T$, and $n=\vert A\vert$.
\end{theorem}
\begin{proof}
    The procedure is given as follows: 
    \begin{enumerate}
        \item Compute any reduced stable partition $\Pi=M_0\mathcal{O}_I$ of $I$, for example using the algorithm presented by \citet{tan91_1} and subsequently decomposing even-length cycles.
        \item Construct the instance $I_E=(A', \succ')$ from $I$ by restricting $A$ and $\succ$ to all agents $a_i\notin A(\mathcal{O}_I)$.
        \item Construct the instance $I_T=(A', \succ'')$ from $I_E$ by removing all agent pairs $\{a_i, a_j\} \subseteq A'$ from $\succ'$ whenever there exists an agent $a_r\in A(\mathcal{O}_I)$ such that $a_i \succ_r \mathcal O_I^{-1}(a_r)$ and $a_r \succ_i a_j$.
        \item Enumerate all stable matchings $S(I_T)$ of instance $I_T$, for example using the algorithm by \citet{federegalapprox}.
        \item For every matching $M_i \in S(I_T)$, construct and output the corresponding reduced stable partition $M_i\mathcal{O}_I$. 
    \end{enumerate}
    Steps 1-3 can be carried out in $O(n^2)$ steps using Algorithm \ref{alg:itconstruction}, as a stable partition can be found in linear time using the algorithm by \citet{tan91_1} and converted into a reduced stable partition in linear time by breaking down even-length cycles into transpositions. Furthermore, the first loop takes $O(n)$ time and the second one takes $O(n^2)$ time by the number of agents, assuming constant time lookup for preference comparisons and partners in $\mathcal{O}_I$. Step 4 can be carried out in $O(\vert S(I_T) \vert n_1 + n_1^2)$ steps using the algorithm presented in \citet{federegalapprox} (recall from Definition \ref{def:stablecycles} that $n_1=A(\mathcal{T}$ $\doubleplus$$\mathcal{E})$ and notice that, in our case, by considering reduced stable partitions only, we have $n_1=A(\mathcal{T})$). Depending on the presentation of the output, step 5 can be carried out in $O(\vert S(I_T) \vert n)$ steps under the assumption that we want to print all cycles for every reduced stable partition. As all steps can be carried out sequentially, we arrive at the stated worst-case complexity. 
\end{proof}

\subsection{Structure of Longer Even Cycles}
\label{section:longercycles}

Reduced stable partitions are canonical in the sense that they cannot be broken down into smaller cycles. However, while \citet{tan91_2} shows that any stable cycle of even length greater than two can be decomposed into a stable collection of transpositions, it may be of theoretical interest to look at how longer cycles can be constructed from a collection of transpositions. Naively, we could try any combination of transpositions to construct longer even cycles. However, this would quickly lead to a combinatorial explosion due to the potential number of transpositions involved.

First, we look at the different ways to decompose an even cycle into transpositions.

\begin{lemma}
\label{lemma:atleasttwo}
    Any even cycle $C=(a_{i_1} \; a_{i_2} \; \dots \; a_{i_{2k}})$ of some stable partition $\Pi$ where $k\geq 2$ can be broken into two distinct collections of transpositions $ C_1 = (a_{i_{1}} \; a_{i_{2}})(a_{i_{3}} \; a_{i_{4}}) \dots (a_{i_{2k-1}} \; a_{i_{2k}})$ and $C_2 = (a_{i_{1}} \; a_{i_{2k}})(a_{i_{2}} \; a_{i_{3}}) \dots (a_{i_{2k-2}} \; a_{i_{2k-1}})$ such that both partitions $\Pi_1 = (\Pi\setminus C)$ $\doubleplus$$C_1$ and $\Pi_2 = (\Pi\setminus C)$ $\doubleplus$$C_2$ are stable.
\end{lemma}
\begin{proof}
    By definition of even cycles, we know that $C=(a_{i_1} \; a_{i_2} \; \dots \; a_{i_{2k}})$ is of length $2k\geq 4$. Suppose that $\Pi_1$ or $\Pi_2$ are not stable. Then either condition T1 or T2 must be violated. However, if $\Pi$ satisfies T1, then both $\Pi_1$ and $\Pi_2$ must satisfy T1 because the repartitioned agents are now in transpositions, such that their predecessors equal their successors, trivially satisfying the condition. Thus, suppose (without loss of generality) that although $\Pi$ satisfies T2, $\Pi_1$ does not. Then there exist agents $a_{i_{p}}$, $a_{i_{q}} \in C$ such that we have $a_{i_{p}}\neq a_{i_{q}}$, $a_{i_{q}} \succ_{i_p} \Pi_1^{-1}(a_{i_{p}}) = C_1(a_{i_{p}})$, and $a_{i_{p}} \succ_{i_q} \Pi_1^{-1}(a_{i_{q}}) = C_1(a_{i_{q}})$. However, we know that for all agents $a_{i_{r}}\in C$, we have $C_1(a_{i_{r}})\succeq_{i_r} C^{-1}(a_{i_{r}}) = \Pi^{-1}(a_{i_{r}})$, contradicting stability of $C$. Thus, both $\Pi_1$ and $\Pi_2$ are stable partitions satisfying T1 and T2.
\end{proof}

However, there might be more ways to rearrange the agents of an even cycle into a stable collection of disjoint transpositions. Consider the instance $I$ shown in Example \ref{table:moretranspositions} which admits the stable partition (indicated in the figure with successors in bold and predecessors underlined) $\Pi = (a_1 \; a_2 \; a_3 \; a_4 \; a_5 \; a_6)$. $I$ also admits the reduced stable partitions $\Pi_1 = (a_1 \; a_2) (a_3 \; a_4)(a_5 \; a_6)$, $\Pi_2 = (a_1 \; a_6) (a_2 \; a_3)(a_4 \; a_5)$, and $\Pi_3 = (a_1 \; a_4) (a_2 \; a_5)(a_3 \; a_6)$, corresponding to three stable matchings of $I$, where $\Pi_3$ is indicated in the figure through circles.

\begin{table}[!htb]
\centering
\caption{{\sc sr} instance with at least five stable partitions}
    \begin{tabular}{ c | c c c c }
    $a_1$ & \circled{$a_4$} & $\mathbf{a_2}$ & \underline{$a_6$} & \dots \\
    $a_2$ & $\mathbf{a_3}$ & \underline{$a_1$} & \circled{$a_5$} & \dots \\
    $a_3$ & \circled{$a_6$} & $\mathbf{a_4}$ & \underline{$a_2$} & \dots \\
    $a_4$ & $\mathbf{a_5}$ & \underline{$a_3$} & \circled{$a_1$} & \dots \\
    $a_5$ & \circled{$a_2$} & $\mathbf{a_6}$ & \underline{$a_4$} & \dots \\
    $a_6$ &  $\mathbf{a_1}$ & \underline{$a_5$} & \circled{$a_3$} & \dots  
    \end{tabular}
\label{table:moretranspositions}
\end{table}

While $\Pi_1$ and $\Pi_2$ can be derived from $\Pi$ by Lemma \ref{lemma:atleasttwo}, $\Pi_3$ is different in the way that it does not preserve the structure of $\Pi$. Note, though, that $I$ also admits the stable partition $\Pi_4 = (a_1 \; a_4 \; a_3 \; a_6 \; a_5 \; a_2)$, which can be broken into $\Pi_1$ and $\Pi_3$ in a similar manner. In fact, when combining two such collections of transpositions $\Pi_a$ and $\Pi_b$ with empty intersection into one large cycle $C$, there are just two options, only one of which can be correct. Start with an arbitrary agent, say $a_1$, and note that either $C(a_1)=\Pi_a(a_1)$ and $C^{-1}(a_1)=\Pi_b(a_1)$, or $C(a_1)=\Pi_b(a_1)$ and $C^{-1}(a_1)=\Pi_a(a_1)$. From there, just alternate between the stable partitions $\Pi_a$ and $\Pi_b$ and continue until the cycle is complete. Note, though, that although one of the constructions must satisfy the stability condition T1, condition T2 might not be satisfied. Consider, for example, the cycle $(a_1 \; a_4 \; a_5 \; a_2 \; a_3 \; a_6)$ produced by joining $\Pi_3$ and $\Pi_2$, which satisfies T1 but not T2, as, for example agents $\{a_1, a_2\}$ prefer each other over their predecessors. 

\begin{theorem}
\label{theorem:tworeduced}
    Any non-reduced stable partition $\Pi$ can be constructed from two reduced stable partitions $\Pi_a$, $\Pi_b$.
\end{theorem}
\begin{proof}
    Consider an arbitrary stable partition $\Pi=\mathcal T$ $\doubleplus$$\mathcal E$ $\doubleplus$$\mathcal{O}_I$, where $\mathcal T$ are the transpositions, $\mathcal E$ are the even cycles, and $\mathcal{O}_I$ are the odd cycles. Then we have shown in Lemma \ref{lemma:atleasttwo} that we can break down each cycle $C\in \mathcal E$ into two different collections of transpositions. Now we can consider the reduced stable partitions $\Pi_a$ and $\Pi_b$, where in $\Pi_a$ every even cycle in $\mathcal E$ is broken down in one way and in $\Pi_b$ the other way. Then precisely $\Pi_a \cap \Pi_b=\mathcal{T}$ $\doubleplus$$\mathcal{O}_I$, so the symmetric difference $D = \Pi_a \triangle \Pi_b$ contains precisely the transpositions formed by breaking down each cycle $C\in\mathcal E$ in two different ways.

    Now in combination with Lemma \ref{lemma:atleasttwo}, we know that the symmetric difference is precisely the set $D=\{(a_{i_j} \; a_{i_{j+1}}) \; \vert \; (a_{i_0} \; a_{i_1} \; a_{i_2} \; \dots \; a_{i_{2k-1}})\in \mathcal E \; \wedge \; 0 \leq j\leq 2k-1 \}$, where addition is taken modulo $2k$ to account for the cyclic structure. Thus, knowing $D$, we have overlapping transpositions such that every agent $a_{i_j}$ is contained in exactly two transpositions $(a_{i_j} \; a_{i_{j+1}})$ and $(a_{i_{j-1}} \; a_{i_j})$. We can now reconstruct $\mathcal E$ cycle by cycle by simply aligning the transpositions as informally shown below Example \ref{table:moretranspositions}, where $a_{i_j}$'s preferences will uniquely determine the alignment by determining whether $a_{i_{j-1}}$ or $a_{i_{j+1}}$ must be the successor to satisfy stability condition T1.

    The result follows as $\Pi = (\Pi_a \cap \Pi_b)$ $\doubleplus$$\mathcal E$, where $\mathcal E$ is constructed from $\Pi_a \triangle \Pi_b$.
\end{proof}

This gives a naive way to construct all stable partitions from the set of reduced stable partitions.

\begin{corollary}
\label{corollary:allpartitionsfromreduced}
    We can construct all stable partitions, $P(I)$, of an {\sc sr} instance $I$ from its set of reduced stable partitions, $RP(I)$, in $O(\vert RP(I)\vert^2 n^2)$ time without repetitions.
\end{corollary}
\begin{proof}
    Knowing that any stable partition can be constructed from two reduced stable partitions due to Theorem \ref{theorem:tworeduced}, we can conclude that in addition to all reduced stable partitions $RP(I)$, we can find all other stable partitions admitted by $I$ by considering every one of the $\binom{\vert RP(I)\vert}{2} = O(\vert RP(I)\vert^2)$ pairs of reduced stable partitions, constructing each candidate partition and verifying in $O(n^2)$ time that it satisfies T1 and T2.

    This process will not generate repetitions, as two distinct pairs of reduced stable partitions will never generate the same non-reduced stable partition candidate. To see this, let $\{\Pi_a, \Pi_b\}, \{\Pi_c, \Pi_d\}$ be two distinct pairs of reduced stable partitions where agent $a_{i_1}$ is in a different transposition in each of them. Then, when merging $\Pi_a$ and $\Pi_b$ to a non-reduced stable partition candidate $\Pi$ using the method above, we have either $\Pi(a_{i_1})=\Pi_a(a_{i_1})$ or  $\Pi(a_{i_1})=\Pi_b(a_{i_1})$ and similar when merging $\Pi_c$ and $\Pi_d$ to candidate $\Pi'$. However, we know $\{\Pi_a(a_{i_1}), \Pi_b(a_{i_1})\}\cap\{\Pi_c(a_{i_1}), \Pi_d(a_{i_1})\}=\varnothing$, so clearly $\Pi\neq \Pi'$.
\end{proof}

However, as $\vert RP(I)\vert$ may be exponential in $n$ and we have seen that many combinations of reduced stable partitions do not lead to a non-reduced stable partition, we focus the study on stable cycles themselves and how they could be useful in enumerating all stable partitions directly.

First, to refine the search for longer even cycles, we can show restrictions on the combinatorics of agents. We want to show that it suffices to know a single predecessor-successor pair of agents in a partial even cycle  $C=(a_{i_1} \; a_{i_2} \; \dots)$ to find a completion if one exists, and that we can find such a completion in polynomial time. For this, we will prove properties about the correspondence between $C$ and $C\setminus \{a_{i_2}\}$, where the latter is a potential odd cycle of some related instance. Specifically, we propose the method outlined in Algorithm \ref{alg:isconstruction} to transform the instance $I$ appropriately to an instance $I_S$ such that it admits the desired invariant odd cycle if it exists. The procedure takes as input the instance and a predecessor-successor pair $a_{i_1}, a_{i_2}$ of the partial (not necessarily stable) cycle in question and creates an instance $I_S$ related to $I$. $I_S$ is constructed by removing agent $a_{i_2}$ and modifying the preference lists of each agent $a_r$ preferred by $a_{i_2}$ to $a_{i_1}$ in $I$ in such a way that $a_{i_1}$ takes the spot of $a_{i_2}$, and its previous occurrence of $a_{i_1}$ in $a_r$'s preference list is deleted. Furthermore, each such agent $a_r$, if not preferred by $a_{i_1}$ to $a_{i_2}$, is promoted in the preference list of $a_{i_1}$ in $I_S$ to precede $a_{i_2}$, while maintaining the relative preference order given by $a_{i_2}$ in $I$.

\begin{algorithm}[!htb]
%\floatname{algorithm}{Procedure}
\renewcommand{\algorithmicrequire}{\textbf{Input:}}
\renewcommand{\algorithmicensure}{\textbf{Output:}}

    \begin{algorithmic}[1]

    \Require{$I$: an {\sc sr} instance; $a_{i_1}$: an agent of some partial (not-necessarily stable) cycle of $I$; $a_{i_2}$: the successor of $a_{i_1}$ in the partial cycle in question}
    \Ensure{$I_S$ : an instance in which $a_{i_1}, a_{i_2}$ are in an odd cycle if the partial cycle is stable}

    \State $I_S \gets I$ (make a copy)

    \For{{agent $a_{r}$ preferred by $a_{i_2}$ over $a_{i_1}$}}
        \State {delete $a_{i_1}$ from preference list of $a_{r}$ in $I_S$}
        \State {replace $a_{i_2}$ by $a_{i_1}$ in preference list of $a_{r}$ in $I_S$}
    \EndFor

    \State {\tt k} $\gets$ {rank}$_{i_1}(a_{i_2})$ 

    \For{{agent $a_{r}$ preferred by $a_{i_2}$ (in order) over $a_{i_1}$ but not by $a_{i_1}$ over $a_{i_2}$}}
        \State {move $a_r$ up in the preference list of $a_{i_1}$ to rank-position {\tt k} in $I_S$}
        \State {\tt k} $\gets$ {\tt k} $+1$
    \EndFor

    \State {delete $a_{i_2}$ from $I_S$ altogether}
    
    \State\Return{$I_S$}

    \end{algorithmic}
    \caption{{\sf ConstructI}$_S (I$, $a_{i_1}$, $a_{i_2})$, constructs $I_S$ from {\sc sr} instance $I$}
    \label{alg:isconstruction}
\end{algorithm}

\begin{lemma}
\label{lemma:oddimplied}
    Let $I$ be an {\sc sr} instance, let $(a_{i_1} \; a_{i_2} \dots)$ be a partial cycle of $I$, and let $I_S$ be the instance constructed from $I$ using Algorithm \ref{alg:isconstruction}. If $(a_{i_1} \; a_{i_2} \dots)$ has a completion to an even cycle in $I$, then $a_{i_1}$ belongs to an odd cycle in $I_S$.
\end{lemma}
\begin{proof}
    Suppose that $C=(a_{i_1} \; a_{i_2} \; a_{i_3} \dots \; a_{i_{2k}})$ is an even cycle in a stable partition $\Pi$ of $I$ for some agents $a_{i_3} \dots \; a_{i_{2k}}$. We claim that $C'=(a_{i_1} \; a_{i_3} \dots \; a_{i_{2k}})$ is an odd cycle of any stable partition $\Pi'$ in $I_S$. For this, we need to prove that stability conditions T1 and T2 are satisfied. 

    For T1, by stability of $\Pi$, in $I$, we have that $a_{i_1}$ prefers $a_{i_2}$ to $a_{i_{2k}}$ and $a_{i_2}$ prefers $a_{i_3}$ to $a_{i_1}$. Hence, by construction, $a_{i_1}$ prefers $a_{i_3}$ to $a_{i_{2k}}$ in $I_S$. Also, $a_{i_3}$ prefers $a_{i_4}$ to $a_{i_2}$ in $I$ by stability of $\Pi$, so $a_{i_3}$ prefers $a_{i_4}$ to $a_{i_1}$ in $I_S$ since $a_{i_2}$ prefers $a_{i_3}$ to $a_{i_1}$ in $I$. 

    For T2, suppose that $\{a_r, a_s\}$ block $\Pi'$ in $I_S$. First, suppose that $a_r=a_{i_1}$. Then $a_s\neq a_{i_3}$ as $(\Pi')^{-1}(a_{i_3})=a_{i_1}$ and $a_{i_1}$ prefers $a_s$ to $a_{i_{2k}}$ in $I_S$. Suppose that $a_{i_2}$ prefers $a_s$ to $a_{i_1}$ in $I$. Then $a_s$ replaced $a_{i_2}$ with $a_{i_1}$ in $I_S$. This means that $\{a_{i_2}, a_{i_s}\}$ blocks $\Pi$ in $I$, a contradiction. Therefore, $a_{i_2}$ does not prefer $a_s$ to $a_{i_1}$ in $I$, so $a_s$ remains in the same position in the preference list of $a_{i_1}$ in $I_S$ as it was in $I$, so $\{a_{i_1}, a_s\}$ blocks $\Pi$ in $I$, a contradiction. 
    
    Now suppose that $a_r=a_{i_3}$. Then $a_{i_3}$ prefers $a_s$ to $a_{i_1}$ in $I_S$. We know that $a_{i_2}$ prefers $a_{i_3}$ to $a_{i_1}$ in $I$. Hence, by construction of $I_S$, $a_{i_3}$ prefers $a_s$ to $a_{i_2}$. Also since $a_s\neq a_{i_1}$ and $a_s\neq a_{i_3}$, it must be the case that $(\Pi')^{-1}(a_s)=\Pi^{-1}(a_s)$, so $\{a_{i_3}, a_s\}$ block $\Pi$ in $I$, a contradiction. 
    
    Lastly, if $\{a_r, a_s \} \cap \{a_{i_1}, a_{i_3}\} = \varnothing$, then clearly $\{a_r, a_s\}$ block $\Pi$ in $I$ as $(\Pi')^{-1}(a_r)=\Pi^{-1}(a_r)$, $(\Pi')^{-1}(a_s)=\Pi^{-1}(a_s)$, $a_r$ prefers $a_s$ to $\Pi^{-1}(a_r)$ and $a_s$ prefers $a_r$ to $\Pi^{-1}(a_s)$.
\end{proof}

Due to odd cycles being invariant, we can use the statement above to reason about the structural restrictions of longer even cycles implied by a predecessor-successor pair.

\begin{lemma}
\label{lemma:distinctfixed}
    If $C=(a_{i_1} \; a_{i_2} \; a_{i_3} \; \dots \; a_{i_{2k}})$ and $C'=(a_{i_1} \; a_{i_2} \; a_{i_3}' \; \dots \; a_{i_{2l}}')$ are both even cycles of $I$, then $k=l$ and $a_{i_j} = a_{i_j}'$ for $3\leq j\leq 2k$.
\end{lemma}
\begin{proof}
    Consider even cycles $C$, $C'$ of $I$. Then by Lemma \ref{lemma:oddimplied}, we must have stable odd cycles $(a_{i_1} \; a_{i_3} \; \dots \; a_{i_{2k}})$ and $(a_{i_1} \; a_{i_3}' \; \dots \; a_{i_{2l}}')$ in the instance $I_S$ using the transformation from $I$ in Algorithm \ref{alg:isconstruction}. However, because $a_1$ is in both cycles and odd cycles are invariant, they are the same cycles, and the result follows.
\end{proof}

Now note that this also implies that no instance can admit, for example, the following two stable partitions: $\Pi = (a_1 \; a_2 \; a_3 \; a_4 \; a_5 \; a_6)$, $\Pi' = (a_1 \; a_2 \; a_3 \; a_4) (a_5 \; a_6)$. 

\begin{lemma}
    If some stable partition $\Pi$ of {\sc sr} instance $I$ contains an even cycle $C=(a_{i_1}\;\dots\; a_{i_k})$, then $C$ cannot be contained in a longer even cycle $C'=(a_{i_1}\;\dots\; a_{i_k} \; a_{j_1}\;\dots\; a_{j_l})$ of some other stable partition $\Pi'$ of $I$.
\end{lemma}
\begin{proof}
    We can consider any predecessor-successor pair in $C$ to see that Lemma \ref{lemma:distinctfixed} implies that there cannot exist such a $C'$.
\end{proof}

The results also greatly restrict the number of even cycles that an instance can admit.

\begin{lemma}
    \label{lemma:countevencycles}
    Any {\sc sr} instance with $n$ agents admits at most $O(n_1^2 + n_2)=O(n^2)$ cycles of length not equal to 2, where $n_2 = \vert A(\mathcal{O}_I)\vert$ and $n_1=n-n_2$.
\end{lemma}
\begin{proof}
    As odd cycles are invariant and any agent can be in at most one odd cycle, there are at most $n_2$ odd cycles.

    Due to Lemma \ref{lemma:distinctfixed}, any candidate predecessor-successor pair is part of at most one even cycle, and hence any pair of agents can be consecutively ordered in at most two even cycles. Thus, we have at most $O(n_1^2)$ even cycles of length not equal to 2.
\end{proof}

Finally, Lemma \ref{lemma:distinctfixed} also establishes the following powerful result.

\begin{theorem}
\label{theorem:completetwo}
    Let $(a_{i_{1}}\; a_{i_{2}} \dots)$ be a partial cycle in some {\sc sr} instance $I$. In linear time, we can determine whether this partial cycle has a completion $C=(a_{i_{1}} \; a_{i_{2}} \; a_{i_{3}} \dots a_{i_{2k}})$ to an even cycle of some stable partition of $I$, or report that no such completion exists. If a completion $C$ to a stable even cycle does exist, then $C$ must be unique and can be found in linear time.
\end{theorem}
\begin{proof}
    We can construct $I_S$ from $I$ using the Algorithm \ref{alg:isconstruction} and find some stable partition $\Pi'$ of $I_S$, all in linear time. Now, if $a_{i_{1}}$ does not belong to an odd cycle $C'$ of $\Pi'$, then by Lemma \ref{lemma:oddimplied}, $(a_{i_{1}} \; a_{i_{2}} \dots)$ has no completion to an even cycle of $I$.

    Otherwise, suppose that $C'=(a_{i_{1}} \; a_{i_{3}} \dots a_{i_{2k}})$ is an odd cycle of $\Pi'$. Then the cycle $C=(a_{i_{1}} \; a_{i_{2}} \; a_{i_{3}} \; \dots \; a_{i_{2k}})$ is an even cycle of $I$ unless $\{a_{i_{1}}, a_{i_{3}}\}$ block in $(\Pi'\setminus C')$ $\doubleplus$$C$.     
    
    Thus, we can either output $C$, which is unique by Lemma \ref{lemma:distinctfixed}, or it must be the case that ``no" is the correct answer, as by the same Lemma, if $(a_{i_1} \; a_{i_2} \dots)$ has a completion to an even cycle in $I$, then $C=(a_{i_1} \; a_{i_2} \; a_{i_3} \dots a_{i_{2k}})$ is the only candidate.
\end{proof}

We have shown above that if $C$ is a partial even cycle of the considered instance $I$, then we must have a corresponding invariant odd cycle $C'$ of some related instance of $I$, and that therefore, $C$ has a unique completion. Furthermore, Theorem \ref{theorem:completetwo} establishes that we only need to consider a potential successor-predecessor pair to see whether that pair is indeed part of some longer unique stable even cycle, and we can do so in linear time.

Furthermore, given an instance $I$ and a (non-partial) candidate cycle $C$, we can also efficiently verify whether it is in some stable partition as it is, by Theorem \ref{theorem:completetwo}, sufficient to consider any successor-predecessor pair in it. If so, we can also explicitly construct a stable partition $\Pi$ containing $C$ using the method described in the proof, where $\Pi = (\Pi'\setminus C')$ $\doubleplus$$C$.

\begin{corollary}
\label{corollary:verify}
    Given an {\sc sr} instance $I$ and a candidate (not necessarily stable) cycle $C$ of even length, we can verify whether $C$ is stable and, if so, construct a stable partition $\Pi$ containing $C$, all in linear time.
\end{corollary}

\subsection{All Stable Cycles}
\label{section:allcycles}

With stable matchings, we can consider the problem of finding all stable pairs, which is defined as the union of all pairs contained in some stable matching of an {\sc sr} instance $I$,
$$SP(I) = \{\{a_i,a_j\}\in M \; \vert \; M\in S(I)\}$$
where $S(I)$ is the set of all stable matchings admitted by $I$. Similarly, one can ask about the fixed pairs, which are the pairs contained in all stable matchings of $I$. Analogously, we introduce the corresponding questions for stable partitions: given an {\sc sr} instance $I$, what are all the stable cycles, denoted $SC(I)$, admitted by some stable partition of the instance? Similarly, what are the (fixed) stable cycles, denoted $FC(I)$, admitted by all its stable partitions?

We have seen how to enumerate all reduced stable partitions of an instance and also how to reconstruct longer even cycles from transpositions, allowing us to enumerate all stable partitions, not necessarily reduced. However, the enumeration of reduced stable partitions might naturally run in time exponential in the number of agents, similar to the enumeration of all stable matchings of a solvable instance. In contrast, when only looking at stable pairs rather than matchings, \citet{federegalapprox} presented an $O(n^3)$ algorithm to find them in a {\sc sr} instance with $n$ agents. Therefore, it is natural to ask whether we can construct a similar efficient algorithm for finding all stable cycles of a stable partition. 

Given an {\sc sr} instance $I$, we are looking to find stable cycles 
$$SC(I) = \{C\in \Pi \; \vert \; \Pi \in P(I) \}$$
where $P(I)$ are all stable partitions admitted by $I$. Now note that the odd cycles $\mathcal{O}_I$ are invariant and can be computed efficiently, so we need to find the union of all stable cycles of even length. Observe that 
$$\mathcal{O}_I \subseteq FC(I) \subseteq RSC(I) \subseteq SC(I),$$
where $RSC(I)$ are the reduced stable cycles (i.e. stable transpositions and odd cycles) of $I$. This is because for any stable cycle $C$ of even length longer than 2, we can decompose it into two collections of stable transpositions by Lemma \ref{lemma:atleasttwo}, so $C$ cannot be fixed. Note that if $\vert P(I) \vert = 1$, then the second and third inclusions hold with equality and if, furthermore,  $\vert A(\mathcal{O}_I)\vert=n$, then the first inequality also holds with equality. Thus, we can find all reduced fixed cycles directly in $O(n^2)$ time by constructing the odd cycles $\mathcal{O}_I$, which are fixed, and the sub-instance $I_T$ (computed using Algorithm \ref{alg:itconstruction}). Then, we can directly derive the fixed pairs (i.e. fixed stable transpositions) as proven by \citet{gusfield89}. 

\begin{theorem}
\label{theorem:findallfc}
    Given an {\sc sr} instance $I$ with $n$ agents, we can find all its fixed cycles in linear time. 
\end{theorem}

Theorem \ref{theorem:correspondence} states that any reduced stable partition $\Pi_i$ is the union of some stable matching $M_i$ of the sub-instance $I_T$ of $I$ (constructed by Algorithm \ref{alg:itconstruction}), and $\mathcal{O}_I$. Therefore, in order to find all stable transpositions, we can apply the known all stable pairs algorithm to $I_T$. 

\begin{lemma}
\label{lemma:reducedcycles}
    Given an {\sc sr} instance $I$ with $n$ agents, we can find all reduced stable cycles, $RSC(I)$, i.e. odd cycles and stable transpositions, in $O(n^2 + n_1^3)$ time, where $n_1=n-\vert A(\mathcal{O}_I)\vert$.
\end{lemma}
\begin{proof}
    We know that we can find all stable pairs of a solvable {\sc sr} instance with $n_1$ agents in $O(n_1^3)$ time using the algorithm by \citet{federegalapprox}. 

    Consider any pair $p$ produced by the algorithm. Then $p$ is a stable pair of some stable matching $M$ of $I_T$, so by Theorem \ref{theorem:correspondence}, $M\mathcal{O}_I$ is a reduced stable partition of $I$ and thus $p$ is a stable transposition. Similarly, if $t$ is some stable transposition of $I$, then there exists some reduced stable partition $\Pi_i$ containing $t$. By Theorem \ref{theorem:correspondence}, there exists some stable matching $M_i$ of $I_T$ such that $\Pi_i=M_i\mathcal{O}_I$. As $t$ is a transposition, we must have that $t\in M_i$, so the stable pairs algorithm must detect and output $t$.

    The complexity follows directly from the size of $I_T$ and the justification by \citet{federegalapprox}.
\end{proof}

This lets us extend Lemma \ref{lemma:countevencycles} to include transpositions.

\begin{theorem}
    Any {\sc sr} instance with $n$ agents admits at most $O(n_1^2+n_2)=O(n^2)$ stable cycles, where $n_2=\vert A(\mathcal{O}_I)\vert$ and $n_1=n-n_2$.
\end{theorem}
\begin{proof}
    By \citet{gusfield89} we know that any solvable instance with $k$ agents admits at most $O(k^2)$ stable pairs, so we will have at most $O(n_1^2)$ stable transpositions. Due to Lemma  \ref{lemma:countevencycles}, there are also at most $O(n_1^2+n_2)$ stable cycles that are not transpositions, so the result follows.
\end{proof}

Now, to compute all stable cycles, not necessarily reduced or fixed, we can use our previous work on reconstructing stable cycles of even length longer than 2 to show the following. 

\begin{theorem}
\label{theorem:findallsc}
    Given an {\sc sr} instance $I$ with $n$ agents, we can find all its stable cycles, $SC(I)$, in $O(n_1^2(n^2+n_1))=O(n^4)$ time, where $n_2=\vert A(\mathcal{O}_I)\vert$ and $n_1=n-n_2$.
\end{theorem}
\begin{proof}
    Due to Lemma \ref{lemma:reducedcycles}, we can find all reduced stable cycles, i.e. all stable transpositions and odd cycles, in $O(n^2 + n_1^3)$ time.

    Now, with regard to the stable cycles of even length longer than 2, we can consider all pairs of agents that could be part of a longer even cycle. There are $\binom{n_1}{2}=O(n^2)$ such subsets in general, but we can limit ourselves to the stable transpositions we found in the previous step due to Theorem \ref{theorem:tworeduced}, in which we showed that any even cycle can be broken into at least two collections of stable transpositions, therefore no agent pair that is not a stable transposition will lead to an even cycle.
    
    For each pair $\{a_i, a_j\}$, we need to consider each of the two orderings $(a_i \; a_j \; \dots )$ and $(a_j \; a_i \; \dots )$. Clearly, there are $O(n_1^2)$ such pairs. Then, for each ordering, we can try to find its unique completion in $O(n^2)$ time due to Theorem \ref{theorem:completetwo} and finally verify that it is a stable cycle, also in $O(n^2)$ time due to Corollary \ref{corollary:verify}. Altogether, we consider $O(n_1^2)$ candidates and need $O(n^2)$ time for each, therefore requiring $O(n_1^2n^2)$ time. As any consecutively ordered pair of agents will only be part of at most one stable cycle of length longer than 2 by Theorem \ref{theorem:completetwo}, we can eliminate all such pairs from the candidate set once detected. This will ensure that the even cycles are enumerated exactly once.

    Altogether, the steps take $O(n_1^2(n^2+n_1))$, or $O(n^4)$ time.
\end{proof}

\subsection{Enumerating Stable Partitions}
\label{section:enumeratingpartitions}

Now that we know how to find all stable cycles of our {\sc sr} instance $I$ efficiently, we show that this leads to an alternative approach to enumerating all stable partitions, $P(I)$, compared to the naive method from Corollary \ref{corollary:allpartitionsfromreduced}. In this procedure, we will recursively build up all stable partitions, starting with the invariant odd cycles and adding one additional cycle with each recursive call. Simultaneously, at each recursive step, we fix an arbitrary agent not yet in any cycle of the partial partition and branch on all its cycles (as given by $SC(I)$) that, together with the current partial partition, lead to at least one stable partition. This way, the recursive tree consists of paths leading from the origin (a partial stable partition containing only the odd cycles) to all leaf nodes (all stable partitions).

The procedure {\sf EnumP$_I$} implements the logic described in the previous paragraph to enumerate all stable partitions. The parameter {\tt cycles} contains a partial stable partition and grows with each recursive call until it is a full stable partition. The parameter {\tt agents} contains all agents in $I$ not yet in any cycle of {\tt cycles} and is naturally decreasing in size as {\tt cycles} grows. Finally, the parameter $SC$ contains all cycles of some stable partition of $I$ containing some agent in {\tt agents}, not necessarily compatible with {\tt cycles} (in the sense that not every cycle in $SC$ belongs to a stable partition with {\tt cycles}). Note that by this behaviour, if {\tt agents} is empty, then {\tt cycles} is a stable partition and we print it. If not, then the algorithm proceeds to remove an arbitrary agent $a_i$ from {\tt agents}. Then, the algorithm checks whether there is a unique cycle $C$ containing $a_i$ such that {\tt cycles}  $\doubleplus$$C$ is part of some stable partition of $I$ using the procedure {\sf FC} from Algorithm \ref{alg:fc}. If there is a unique such $C$, then we reduce the set of stable cycles $SC$ to $SC_{red}$ which does not contain any cycles with agents in $C$ and recurse with the increased partial stable partition {\tt cycles}  $\doubleplus$$C$, the reduced set of stable cycles to choose from $SC_{red}$, and the reduced set of remaining agents {\tt agents}$\setminus A(C)$. 

If there is no unique choice $C$, then the procedure loops through all candidate cycles $C'\in SC$ that contain $a_i$, at least two of which must correspond to a partial stable partition {\tt cycles}  $\doubleplus$$C'$ by construction. Whether {\tt cycles}  $\doubleplus$$C'$ does correspond to a partial stable partition or only to a collection of cycles that is not a subset of any stable partition can be determined with the procedure {\sf verify} which is omitted but uses the result from Corollary \ref{corollary:verify} to check whether the collection of cycles completes to a stable partition. If so, we, again, create the appropriate $SC_{red}$ and recurse as above.

When there exists a unique cycle for the current considered agent that needs to be added to the partial partition (our candidate partition, which is known to be a subset of some stable partition), the helper function {\sf FC} will speed up the computation. We can show, furthermore, that if such a unique choice of cycle, which we can refer to as a fixed cycle, exists, it must always be a transposition. To see this, consider an {\sc sr} instance $I$. If we have a collection $S$ of cycles including all odd cycles that are part of some stable partition of $I$ and a single agent $a_i\in A(I)$ not in $A(S)$, then by the existence of stable partitions, there must exist a cycle $C\notin S$ where $a_i\in A(C)$ such that $S$ $\doubleplus$$C$ is part of some stable partition of $I$. Now, if $C$ is unique, then $C$ must be a transposition. Suppose not, then $C$ must be of even length at least 4 (by $S$ containing all odd cycles), but then we can break it into two collections of stable transpositions $C_1,C_2$ such that also $S$ $\doubleplus$$C_1$ and $S$ $\doubleplus$$C_2$ are part of some stable partition. Thus, unique choices must be transpositions. 

Now, the procedure {\sf FC} described in Algorithm \ref{alg:fc} takes such a collection $S$ of cycles as input {\tt cycles}, as well as a single agent $a_i\in A(I)\setminus A(S)$ and generates a sub-instance $I'$ of $I$ consisting only of the agents not in $S$ with their preference lists truncated such that any stable partition $\Pi$ of $I'$ corresponds to a stable partition $\Pi$ $\doubleplus$$S$ of $I$. In fact, $I'$ must be solvable as there must exist some such stable partition by construction, and if any would contain an odd cycle, {\tt cycles} would not contain all odd cycles. Thus, by the argument in the previous paragraph, to check whether the agent $a_i$ is in a unique cycle $C$ such that $S$ $\doubleplus$$C$ is a stable partition of $I$, it suffices to apply the fixed pairs algorithm for stable matchings, for example using the linear time procedure (denoted by {\sf getFixedPairs}) by \citet{gusfield89}, and returning (through an $O(n)$ search) the fixed pair containing $a_i$, if it exists, otherwise returning null to report that no such fixed pair exists (to indicate that the choice is not unique).

\begin{algorithm}[!htb]
%\floatname{algorithm}{Procedure}
\renewcommand{\algorithmicrequire}{\textbf{Input:}}
\renewcommand{\algorithmicensure}{\textbf{Output:}}

    \begin{algorithmic}[1]

    \Require{{\tt cycles}: a collection of cycles which is a subset of some stable partition of $I$ and includes all odd cycles of $I$; $a_i$ : a single agent in $A(I)\setminus A(${\tt cycles}$)$}
    \Ensure{{\tt pair} : a pair of two agents in $A(I)\setminus A(${\tt cycles}$)$, potentially null}

    \State $I' \gets I$ (make a copy of $I$)
    
    \For{{\tt cycle} in {\tt cycles}}
        \For{$a_j$ in {\tt cycle}}
            \For{$a_p$ preferred by $a_j$ over {\tt cycle}$^{-1}(a_j)$ in $I$}
                \State delete $a_j$ and all successors of $a_j$ from $a_p$'s preference list in $I'$ (if they exist)
            \EndFor
            \State delete $a_j$ from the set of agents of $I'$ altogether
        \EndFor
    \EndFor

    \State {\tt fixed} $\gets$ {\sf getFixedPairs}($I'$)
    \If{$\{a_i, a_k\}$ in {\tt fixed} for some $a_k\in A(I')$}
        \State \Return ($a_i$, $a_k$)
    \EndIf
    \State \Return null
    \end{algorithmic}

    \caption{{\sf FC}({\tt cycles}, $a_i$), Fixed Cycle Helper Function for Algorithm \ref{alg:pienum}}
    \label{alg:fc}
\end{algorithm}

Overall, our method to enumerate $P(I)$ first finds all $O(n_1^2+n_2)$ stable cycles, $SC(I)$, in $O(n_1^2(n^2+n_1))$ time as previously shown, and then calling the procedure {\sf EnumP$_I$} described in Algorithm \ref{alg:pienum} with stable cycles $SC(I)\setminus\mathcal{O}_I$ as parameter $SC$, the invariant odd cycles $\mathcal{O}_I$ as parameter {\tt cycles}, and agents $A(I)\setminus A(\mathcal{O}_I)$ as parameter {\tt agents}.

\begin{algorithm}[!htb]
%\floatname{algorithm}{Procedure}
\renewcommand{\algorithmicrequire}{\textbf{Input:}}
\renewcommand{\algorithmicensure}{\textbf{Output:}}
    \begin{algorithmic}[1]

    \Require{$SC$: all cycles of some stable partition of $I$ containing some agent in {\tt agents} and no agent in $A(${\tt cycles}$)$; cycles: a collection of cycles which is a subset of some stable partition of $I$ and includes all odd cycles of $I$; {\tt agents}: all agents in $A(I)\setminus A(${\tt cycles}$)$}
    
    \If{{\tt agents} $=\varnothing$}
        \State {print({\tt cycles})}
        \State \Return
    \EndIf
    
    \State {$a_i \gets $ {\tt agents}.pop()}

    \State {{\tt fixed} $\gets$ {\sf FC}($I$, {\tt cycles}, $a_i)$}

    \If{{\tt fixed} $\neq$ null}
        \State $SC_{red} \gets \{(a_{j_1}\;\dots\; a_{j_k})\in SC \;\vert\; \forall\; 1\leq s\leq k, \; a_{j_s} $ not in {\tt fixed}$\}$
        \State {\sf EnumP}$_I$($SC_{red}$, {\tt cycles}  $\doubleplus${\tt fixed}, {\tt agents}$\setminus A$({\tt fixed}))
    \Else
        \For{{{\tt cycle}} in $SC$ containing $a_i$}
            \If{{\sf verify}({\tt cycles}  $\doubleplus${\tt cycle})}
                \State $SC_{red} \gets \{(a_{j_1}\;\dots\; a_{j_k})\in SC \;\vert\; \forall\; 1\leq s\leq k, \; a_{j_s} ${ not in {\tt cycle}}$\}$
                \State {\sf EnumP}$_I$($SC_{red}$, {\tt cycles}  $\doubleplus${\tt cycle}, {\tt agents}$\setminus A$({\tt cycle}))
            \EndIf
        \EndFor
    \EndIf
    \end{algorithmic}
    \caption{{\sf EnumP$_I$}($SC$, {\tt cycles}, {\tt agents}), recursively enumerates all stable partitions $P(I)$}
    \label{alg:pienum}
\end{algorithm}

\begin{theorem}
\label{theorem:enumpi}
    Let $I$ be an instance of {\sc sr} with $n$ agents. Then Algorithm \ref{alg:pienum} enumerates all stable partitions $P(I)$ of $I$ without repetition in $O(\vert P(I)\vert n^3 + n^4)$ time, or more specifically, $O((\vert P(I)\vert + n_1)n_1n^2+n_1^3)$ time, where $n_2=\vert A(\mathcal{O}_I)\vert$ and $n_1=n-n_2$, when called as described above.
\end{theorem}
\begin{proof}
    To establish correctness, note that by the definition of stable cycles, each stable cycle is contained in at least one stable partition and all stable partitions are made up only of (some, not necessarily all) stable cycles. At every node in the recursive tree, if all previous cycles on the recursive path are compatible, i.e., belong to some stable partition together, then if {\tt agents} is null, all agents are part of some cycle and thus the path represents a stable partition, and we can stop. Otherwise, if {\tt agents} is not null, there must be some cycle or cycles $S$ in the $SC$ parameter that \texttt{cycles}  $\doubleplus$$C$ is a partial stable partition for all $C\in S$. If $\vert S\vert = 1$, then we will detect it using the procedure {\sf FC} and can immediately filter out all cycles in $SC$ from further consideration that contain any agent in $C$. Similarly, if $\vert S\vert >1$, then we will need to consider all cycles in $SC$ that contain $a_i$, check whether it is compatible with \texttt{cycles} and, if so, recurse with a subset of $SC$ not containing any cycles with overlapping agents. This procedure detects every stable partition exactly once, as by construction, every path contains distinct sets of cycles such that every agent is contained exactly once. As the first call passes only the odd cycles to the procedure, it is trivially compatible, as the odd cycles are invariant, and correctness follows by induction.
    
    Now with regards to the computational complexity, note that in the recursive tree, there will be exactly $\vert P(I)\vert$ leaf nodes, $O(\vert P(I)\vert n_1)$ internal nodes with one child because each stable partition will contain $O(n)$ cycles and the odd cycles are given, and $O(\vert P(I)\vert)$ internal nodes with at least two children by properties of binary trees. Every leaf and internal node with one child takes $O(n^2)$ time due to Algorithm \ref{alg:fc} and building $SC_{red}$, and every internal node with more children takes $O(n_1n^2)$ time because we consider at most $O(n_1)$ different stable cycles. This is because each agent can be contained in $O(n_1)$ transpositions and can have at most $O(n_1)$ successors in even cycles, and by Lemma \ref{lemma:distinctfixed} in at most one even cycle per predecessor-successor pair. We noted in the description of Algorithm \ref{alg:fc} that testing whether a candidate set of cycles is part of a stable partition takes $O(n^2)$ time, and building $SC_{red}$ also takes $O(n_1^2)$ time. Thus, the recursive procedure takes, overall, $O(\vert P(I) \vert n_1n^2)$ time, with an additional $O(n_1^2(n^2+n_1))$ overhead to compute all stable cycles once.
\end{proof}

\section{Complexity of Profile-Optimal Stable Partitions}
\label{section:complexity}

We have shown how to efficiently enumerate all reduced stable partitions admitted by an {\sc sr} instance $I$. However, as the size of the set $S(I_T)$ of all stable matchings of the sub-instance $I_T$ of $I$ (constructed using Algorithm \ref{alg:itconstruction}) can, in general, be exponential in the number of agents $n$, this might not be a good solution to finding a reduced stable partition with some special property. Therefore, we might be interested in computing some ``fair" or ``optimal" stable partitions directly. As the odd cycles of {\sc sr} instances are invariant, there is no point in seeking odd cycles with special properties for a given instance. However, we can consider the set of stable partitions of a given instance and seek, for example, a reduced stable partition with minimum regret. In this case, it is not enough to simply leverage the result from Theorem \ref{theorem:correspondence}, as we are deleting entries in the preference lists during the construction of $I_T$, which makes the resulting profile entries in the larger instance inconsistent. For this, we use a \emph{padding method} by replacing deleted agents in the preference lists with dummy agents which are guaranteed to be matched to each other by construction. Therefore, a solvable instance similar to $I_T$ is created in which the ranks are maintained and known algorithms can be applied. This method and its correctness are outlined in Section \ref{section:paddingargument}. Section \ref{section:complexityofproblems} establishes the computational complexity of computing various types of stable partitions directly, and \ref{section:approxofproblems} gives some details about the approximability of the problems that we show to be NP-hard in general.

\subsection{Padding Method}
\label{section:paddingargument}
Whenever, during the construction of $I_E$ and $I_T$, an agent is supposed to be deleted from some preference list, we can instead replace the corresponding preference list entry to be deleted with a dummy agent entry that is guaranteed to be matched to another dummy agent in any stable matching or partition by making them rank each other as their respective first choices. Agents in odd cycles are automatically replaced in all preference lists of agents that are not in odd cycles and can be deleted altogether at the end of the transformation. Clearly, for an instance with $n$ agents, we need at most $n$ dummy agents in the transformation as they can be reused across but not within preference lists. If the number of required dummy agents is odd, we can simply add one more dummy agent to guarantee that all dummy agents are in fixed pairs of any stable matching of the transformed instance. As we are replacing all agents in odd cycles with dummy agents in fixed pairs, we transform the instance $I$ to a solvable instance $I_P$ with at most $2n$ agents, allowing us to apply any of the known algorithms now to find stable matchings satisfying additional optimality criteria with no asymptotic increase in complexity. We denote these $k$ fixed pairs by $D=(d_1 \; d_2)\dots (d_{k-1} \; d_k)$ and can map the solution $M$ in $I_P$ back to a solution $\Pi$ for the original instance $I$ by excluding the dummy agents and including the original odd cycles, i.e., $\Pi=(M\setminus D)$ $\doubleplus$$\mathcal{O}_I$.

The transformation from $I$ to $I_P$ is illustrated in Example \ref{table:dummies} on the instance previously shown in Example \ref{table:bijectionfail}. Two dummy agents, $d_1$ and $d_2$, are introduced, replacing the preference list entries that would otherwise be deleted in the construction of $I_T$. The changes due to the transformation are indicated using dotted circles, and the stable matching corresponding to the one in Example \ref{table:bijectionfail} is indicated using unbroken circles.

\begin{table}[!htb]
\centering
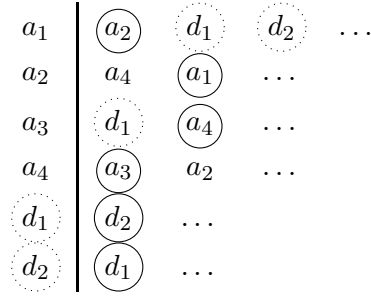
\caption{An instance $I_P$ transformed from $I$ of Example \ref{table:bijectionfail}}
    \begin{tabular}{ c | c c c c }
    $a_1$ & \circled{$a_2$} & \dottedcircled{$d_1$} & \dottedcircled{$d_2$} & \dots \\
    $a_2$ & $a_4$ & \circled{$a_1$} & \dots &  \\
    $a_3$ & \dottedcircled{$d_1$} & \circled{$a_4$} & \dots &  \\
    $a_4$ & \circled{$a_3$} & $a_2$ & \dots &  \\
    \dottedcircled{$d_1$} & \circled{$d_2$} & \dots & & \\
    \dottedcircled{$d_2$} & \circled{$d_1$} & \dots & &  
    \end{tabular}
\label{table:dummies}
\end{table}

The correctness of the correspondence is established by the following Lemma.

\begin{lemma}
    \label{lemma:dummycorrespondence}
    Let $I=(A,\succ)$ be any {\sc sr} instance and let $I_P$ be constructed from $I$ as described above. Then $M$ is a stable matching of $I_P$ if and only if $\Pi=(M\setminus D)$ $\doubleplus$$\mathcal{O}_I$ is a reduced stable partition of $I$, where $D$ is the collection of dummy agent fixed pairs and $\mathcal{O}_I$ are the invariant odd cycles of $I$.
\end{lemma}
\begin{proof}
    If suffices to show that $M$ is a stable matching of $I_P$ if and only if $M\setminus D$ is a stable matching of $I_T$ constructed from $I$ using Algorithm \ref{alg:itconstruction}. Then, by Theorem \ref{theorem:correspondence}, $M\setminus D$ is a stable matching of $I_T$ if and only if $\Pi=(M\setminus D)$ $\doubleplus$$\mathcal{O}_I$ is a reduced stable partition of $I$ and the result follows.
    
    Indeed, if $M$ is a stable matching of $I_P$, then $M\setminus D$ is a matching of $I_T$. Suppose $M\setminus D$ is not stable in $I_T$, then, by the construction of $I_P$, there exist two agents $a_i, a_j\in A\setminus A(\mathcal{O}_I)$ that prefer each other over their partners in $M\setminus D$. However, the preference relations of all agents in $A\setminus A(\mathcal{O}_I)$ are the same for $I_P$ and $I_T$, so this would contradict the stability of $M$ in $I_P$.

    On the other hand, if $M'$ is a stable matching of $I_T$, then $M'\cup D$ is a matching in $I_P$ as we have $A(D)\cap A(M')=\varnothing$. Suppose that $M'\cup D$ is not stable in $I_P$. As the preference relations in $I_T$ for all agents matched in $M'$ remain the same in $I_P$, no two agents matched in $M'$ can block $M'\cup D$ in $I_P$. Also, no dummy agent can block with another dummy agent or an agent matched in $M'$, because all dummy agents are matched to their first choice.
\end{proof}

\subsection{Complexity of Problems}
\label{section:complexityofproblems}

The padding method allows us to establish further results on the correspondence between stable partitions and stable matchings. Specifically, we show that the profile-based optimality criteria for stable matchings carry over naturally to stable partitions, and we can show that the computational complexity results are consistent. 

First, we look at stable partitions with minimum regret. Using the padding method, we can show that the problem of finding a minimum-regret stable partition is tractable. 

\begin{lemma}
    \label{lemma:minregret}
    Let $I$ be an {\sc sr} instance, let $\Pi=M_0\mathcal{O}_I$ be a reduced stable partition of $I$, where $M_0$ is a collection of transpositions and $\mathcal{O}_I$ are the odd cycles in $\Pi$, and $I_P$ be the transformed instance using the padding method. If $M_0\cup D$ is a stable matching of $I_P$ with minimum regret (where $D$ is the part of the matching involving dummy agents), then $\Pi$ is a reduced stable partition with minimum regret.
\end{lemma}
\begin{proof}
    Lemma \ref{lemma:dummycorrespondence} states that if $M_0$ is a minimum-regret stable matching of $I_P$, then $\Pi=M_0$ $\doubleplus$$\mathcal{O}_I$ is a reduced stable partition of $I$.
    
    Now suppose that there exists a reduced stable partition $\Pi' \neq \Pi$ such that $r(\Pi')<r(\Pi)$. Then again, by Lemma \ref{lemma:dummycorrespondence}, we must have $\Pi'= M_1\mathcal{O}_I$, where $M_1\cup D\neq M_0\cup D$ is some stable matching of $I_P$. Note that the dummy agents do not affect the regret of the partition, as all of them are matched to their first choice by construction. Now, if $r(\Pi')<r(\Pi)$, then we must have that $r(M_1)<r(M_0)$ by invariance of $r(\mathcal{O}_I)$. This contradicts that $M_0$ is a minimum-regret stable matching of $I_P$, proving the statement.
\end{proof}

Note that the converse does not necessarily hold, because a reduced stable partition $\Pi=M_0\mathcal{O}_I$ as above could have $r(M_0)<r(\Pi)=r(\mathcal{O}_I)$.

We can furthermore show that there is no gap in minimum regret between the set of reduced stable partitions and the set of stable partitions.

\begin{lemma}
\label{lemma:reducedregret}
    Let $\Pi$ be a minimum-regret reduced stable partition of some {\sc sr} instance $I$ and let $r(\Pi)$ be its regret. Then for any stable partition $\Pi'$ of $I$, we have $r(\Pi)\leq r(\Pi')$.
\end{lemma}
\begin{proof}
    Suppose not, then there exists some non-reduced minimum regret stable partition $\Pi'$ of instance $I$ which has smaller regret $r_1 = r(\Pi')$ than the reduced stable partition $\Pi$ with minimum regret $r_2 = r(\Pi)$ over all reduced stable partitions.

    The odd cycles are invariant for all stable partitions, so $r_2>r_1$ must come from choosing transpositions over longer even cycles. However, we can show that breaking down an even cycle $C=(a_{i_{1}} \; a_{i_{2}} \; \dots \; a_{i_{2k}})$ preserves its regret, i.e. that $r(C)=r(C')$, where $C'=(a_{i_{1}} \; a_{i_{2}})\dots (a_{i_{2k-1}} \; a_{i_{2k}}))$. This holds because any agent's predecessor in $C'$ is at least as good as their predecessor in $C$ due to stability condition T1. As this transformation does not impact stability due to Lemma \ref{lemma:atleasttwo} and because we can break down every even cycle larger than 2 like this, some reduced stable partition can achieve the same regret as any non-reduced stable partition, therefore proving the statement.
\end{proof}

\begin{theorem}
\label{theorem:minregret}
    A stable partition with minimum regret always exists and can be computed in linear time.
\end{theorem}
\begin{proof}
    \citet{gusfield89} already showed that we can compute a stable matching with minimum regret in linear time for solvable instances. Thus, we can compute any reduced stable partition, derive the solvable instance $I_P$ from any {\sc sr} instance $I$, find a stable matching with minimum regret, and map it back to the instance $I$, which produces the correct result due to Lemma \ref{lemma:minregret}. Lemma \ref{lemma:reducedregret} states that the minimum regret reduced stable partition is also a minimum regret stable partition. Existence follows from the existence of stable partitions and all computation and transformation steps can be performed in linear time.
\end{proof}

Consistent with stable matchings, minimum regret remains an exception concerning the tractability of profile-based optimal stable partitions. We now also consider the following decision problems.

\begin{definition}
    Given an {\sc sr} instance $I$ with $n$ agents and regret $r$, an integer $k$, and a vector $\sigma\in\{0,1,\dots, n\}^n$, let 
    \begin{itemize}
        \item {\sc FC-Dec-SR-SP} denote the problem of deciding whether $I$ admits a stable partition $\Pi$ where $p_1(\Pi)\geq k$,
        \item {\sc Rank-Dec-SR-SP} denote the problem of deciding whether $I$ admits a stable partition $\Pi$ where $p(\Pi)\succeq \sigma$,
        \item {\sc RM-Dec-SR-SP} denote the problem of deciding whether $I$ admits a minimum regret stable partition with at most $k$ $r$th choices,
        \item {\sc Gen-Dec-SR-SP} denote the problem of deciding whether $I$ admits a stable partition $\Pi$ where $\sigma^{rev} \succeq p^{rev}(\Pi)$, and
        \item {\sc Egal-Dec-SR-SP} denote the problem of deciding whether $I$ admits a stable partition with cost at most $k$.
    \end{itemize}
\end{definition}

For all of these problems, we show that no non-reduced stable partition is ``better" with regards to our optimality criteria than the best reduced stable partition, which allows us to only consider reduced stable partitions. Knowing the NP-completeness of the associated problems above for stable matchings, we can show that all of these problems are NP-complete for stable partitions, even if the instance is guaranteed to be solvable. One might ask whether the same result holds for unsolvable instances, which have some different structural properties. Using a reduction from solvable instances, we can show that this is indeed the case.

First, to show the next statements, the following Lemma establishes that no non-reduced stable partition satisfies more first-choices than the maximum number of first-choices satisfied by any reduced stable partition. 

\begin{lemma}
\label{lemma:firstchoicereduced}
    Let $C=(a_{i_1} \; a_{i_2} \;\dots\; a_{i_k})$ be an even cycle with $k\geq 4$ and profile $p(C)$. Then both collections of transpositions $C_1 = (a_{i_1}\; a_{i_2}) \dots (a_{i_{k-1}} \; a_{i_{k}})$ and $C_2 = (a_{i_1}\; a_{i_{k}}) \dots (a_{i_{k-2}} \; a_{i_{k-1}})$ are stable and either $p_1(C_1)\geq p_1(C)$ or $p_1(C_2)\geq p_1(C)$.
\end{lemma}
\begin{proof}
    By Lemma \ref{lemma:atleasttwo}, both $C_1$ and $C_2$ are stable. Now for $a_{i_j}\in C$, we know that $C(a_{i_j}) \succ_{i_j} C^{-1}(a_{i_j})$ by stability condition T1 and the fact that $C$ is longer than a transposition, therefore proving strict inequality. Thus, any agent $a_{i_j}$ that contributes a first choice in $C$ must rank their successor as a first choice, and no agent in $C$ can have their predecessor ranked as a first choice. Therefore, any such agent $a_{i_j}$ adds exactly 1 to $p_1(C)$. Now, if we break $C$ into transpositions, an agent $a_{i_j}$ who is paired with their first choice successor now has their first choice as successor and predecessor simultaneously, thus adding 2 either to $p_1(C_1)$ or to $p_1(C_2)$. Now, for any two agents $a_{i_r}, a_{i_s}$ that get their first choice in $C$ and have an even number of agents between them, one will get their first choice assigned in $C_1$ and one in $C_2$. Any two such agents that have an odd number of agents between them will both get their first choice assigned in $C_1$ or $C_2$. Hence, at least one of $C_1$ and $C_2$ will achieve at least as many first choices as in $C$.
\end{proof}

This lets us limit the search for first-choice maximal stable partitions on reduced stable partitions, as we can always break down longer even cycles into transpositions efficiently without loss of first-choice maximality.

\begin{lemma}
\label{lemma:firstchoice}
    {\sc FC-Dec-SR-SP} is NP-complete, even if the given instance is guaranteed to be solvable. 
\end{lemma}
\begin{proof}
    Note that {\sc FC-Dec-SR-SP} is in NP, as given instance $I$, integer $k$, and partition $\Pi$, we can check whether $\Pi$ is a stable partition of $I$ and verify that it has at least $k$ first-choices, all in linear time.

    \citet{CooperPhD} proved that {\sc FC-Dec-SR-SM}, the problem of deciding whether a solvable {\sc sr} instance $I$ with complete preference lists admits a stable matching with at least $k$ first choices, is NP-complete. To establish our result, we will show that for a solvable instance $I$ with complete preference lists, $I$ admits a stable matching $M$ where $p_1(M)\geq k$ if and only if $I$ admits a stable partition $\Pi$ where $p_1(\Pi)\geq 2k$.

    Suppose that $I$ admits a stable matching $M=\{\{a_{i_1},a_{i_2}\}\dots \{a_{i_{n-1}}, a_{i_{n}}\}\}$ with at least $k$ first choices, i.e. at least $k$ agents are assigned a partner that is ranked first in their preference list. Then, immediately, $M$ has an associated stable partition $\Pi = (a_{i_1}\; a_{i_2})\dots(a_{i_{n-1}}\; a_{i_{n}})$, in which at least $k$ agents have a successor that is their first-choice. As all cycles are transpositions, their predecessors are also their successors, so in the associated profile $p(\Pi)$, we will have $p_1(\Pi)\geq 2k$ as desired.

    For the other direction, suppose that $I$ admits a stable partition $\Pi$ with at least $2k$ first choices in $p(\Pi)$. By Lemma \ref{lemma:firstchoicereduced} we can break down $\Pi$ into a reduced stable partition $\Pi'$ if $\Pi$ is not reduced, such that $p_1(\Pi')\geq p_1(\Pi)\geq 2k$. Because we assume that $I$ is solvable, by Theorem \ref{theorem:correspondence}, $\Pi'$ corresponds to a stable matching $M$ of $I$ and so $p_1(M)=\frac{p_1(\Pi')}{2}\geq k$ as desired.
\end{proof}

\begin{theorem}
\label{theorem:fchardness}
    Let $I$ be an {\sc sr} instance and $k$ an integer. {\sc FC-Dec-SR-SP}, the problem of deciding whether $I$ admits a stable partition with at least $k$ first choices, is NP-complete, regardless of whether $I$ is guaranteed to be solvable or unsolvable.
\end{theorem}
\begin{proof}
    If $I$ is solvable, the result follows immediately from Lemma \ref{lemma:firstchoice}. 
    
    Now suppose that we can, in polynomial time, decide {\sc FC-Dec-SR-SP} for an unsolvable instance $I$. Then we could also solve the problem for solvable instances, contradicting the above. Specifically, take any solvable instance $I_1$ with $n$ agents and transform it to an unsolvable instance $I_2$ by adding six agents $a_{n+1},\dots, a_{n+6}$ with preferences given in Example \ref{table:firstchoice3cycles}, where $a_{n+i}=b_{i}$, such that they form two invariant three-cycles $(b_1\;b_2\;b_3)(b_4\;b_5\;b_6)$ in any stable partition without satisfying any first-choices. One can easily verify the construction by checking that each agent prefers their successor over their predecessor and no two agents prefer each other over their predecessors. As the three-cycles are of odd length, they must be invariant. Finally, $I_2$ is clearly unsolvable due to the new odd cycles, but any stable partition $\Pi_2$ of $I_2$ corresponds to a stable partition $\Pi_1 = \Pi_2 \setminus \{(a_{n+1} \; a_{n+2} \; a_{n+3})(a_{n+4} \; a_{n+5} \; a_{n+6})\}$ of $I_1$ and vice versa. Furthermore, $p_1(\Pi_1)=p_1(\Pi_2)$, so some stable partition of $I_1$ satisfies at least $k$ first-choices if and only if some stable partition of $I_2$ satisfies at least $k$ first-choices.
\end{proof}

\begin{table}[!htb]
    \centering
    \caption{Added 3-cycles for first-choice maximal reduction}
        \begin{tabular}{ c | c c c c }
        $b_{1}$ & $b_{4}$ & \circled{$b_{2}$} & \dottedcircled{$b_{3}$} & \dots \\
        $b_{2}$ & $b_{4}$ & \circled{$b_{3}$} & \dottedcircled{$b_{1}$} & \dots \\
        $b_{3}$ & $b_{5}$ & \circled{$b_{1}$} & \dottedcircled{$b_{2}$} & \dots \\
        $b_{4}$ & $b_{3}$ & \circled{$b_{5}$} & \dottedcircled{$b_{6}$} & \dots \\
        $b_{5}$ & $b_{1}$ & \circled{$b_{6}$} & \dottedcircled{$b_{4}$} & \dots \\
        $b_{6}$ & $b_{1}$ & \circled{$b_{4}$} & \dottedcircled{$b_{5}$} & \dots \\
        \end{tabular}
    \label{table:firstchoice3cycles}
\end{table}

The hardness of finding rank-maximal stable partitions follows directly from the previous result.

\begin{theorem}
\label{theorem:rmcompleteness}
    {\sc Rank-Dec-SR-SP} is NP-complete, regardless of whether the instance is guaranteed to be solvable or unsolvable.
\end{theorem}
\begin{proof}
    Note that {\sc Rank-Dec-SR-SP} is in NP, as given an instance $I$, vector $\sigma$, and partition $\Pi$, we can check whether $\Pi$ is a stable partition of $I$ and check whether $p(\Pi)\succeq \sigma$, all in linear time. 
    
    We showed above in Theorem \ref{theorem:fchardness} that {\sc FC-Dec-SR-SP} is NP-complete for solvable and for unsolvable instances. Now given a problem instance $(I, k)$ to {\sc FC-Dec-SR-SP}, $(I, k)$ is a yes-instance if and only if $(I,(k \;0\dots 0))$ is a yes-instance to {\sc Rank-Dec-SR-SP}. 
\end{proof}

A similar result holds for regret-minimal stable partitions. First, the following Lemma establishes that no non-reduced minimum-regret stable partition with regret $r$ satisfies fewer $r$th choices than any reduced minimum-regret stable partition with minimal $r$th choices satisfies.

\begin{lemma}
\label{lemma:lastchoicereduced}
    Let $C=(a_{i_1} \; a_{i_2}\dots a_{i_k})$ be an even cycle with profile $p$ and regret $r$ of some instance $I$ with $n$ agents. Then both $C_1 = (a_{i_1}\; a_{i_2}) \dots (a_{i_{k-1}} \; a_{i_{k}})$ and $C_2 = (a_{i_1}\; a_{i_{k}}) \dots (a_{i_{k-2}} \; a_{i_{k-1}})$ are stable and either $p_r(C_1)\leq p_r(C)$ and $p_j(C_1)=0$ for all $r<j\leq n$ or $p_r(C_2)\leq p_r(C)$ and $p_j(C_2)=0$ for all $r<j\leq n$.
\end{lemma}
\begin{proof}
    By Lemma \ref{lemma:atleasttwo}, both $C_1$ and $C_2$ are stable. Now, for any $a_{i_j}\in C$, we know $C(a_{i_j}) \succ_{i_j} C^{-1}(a_{i_j})$ by stability condition T1 and the fact that $C$ is longer than a transposition, therefore proving strict inequality. Therefore, because no agent can have a successor or predecessor of rank worse than $r$, any agent present in $C$ must receive a successor of rank strictly better than $r$, and any agent $a_{i_j}$ that receives their $r$th choice in $C$ (either as a predecessor or as a successor) must rank their predecessor in rank $r$ (and not their successor, as T1 implies that then their predecessor would be strictly worse than $r$, a contradiction). Hence, any such agent $a_{i_j}$ contributes exactly 1 to $p_r(C)$. If we break $C$ into transpositions, an agent $a_{i_j}$ who is paired with their $r$th choice predecessor now has them also as their successor, thus adding 2 to either $p_r(C_1)$ or $p_r(C_2)$. Clearly, $p_j(C_1)=p_j(C_2)=0$ for all $r<j\leq n$ is true, as everyone gets a partner in $C_1, C_2$ no worse than their predecessor in $C$. Now, for any two agents $a_{i_l}, a_{i_m}$ that get their $r$th choice predecessor in $C$ and have an even number of agents between them, one will get their $r$th choice assigned in $C_1$ and one in $C_2$. For any two such agents that have an odd number of agents between them, they will both get their $r$th choice assigned in either $C_1$ or $C_2$. Hence, the result follows.
\end{proof}

This, again, lets us consider reduced partitions and use existing complexity results.

\begin{lemma}
\label{lemma:rthchoicecomplete}
    {\sc RM-Dec-SR-SP} is NP-complete, even when the given instance is guaranteed to be solvable. 
\end{lemma}
\begin{proof}
    Note that {\sc RM-Dec-SR-SP} is in NP, as given instance $I$ with minimum regret $r$, integer $k$, and partition $\Pi$, we can check whether $\Pi$ is a minimum regret stable partition of $I$ and verify that it assigns at most $k$ $r$th choices, all in linear time.

    \citet{CooperPhD} proved that {\sc RM-Dec-SR-SM}, the problem of deciding whether a solvable {\sc sr} instance $I$ admits a stable matching with at most $k$ $r$th-choices and 0 choices higher than $r$, is NP-complete. Note that this is equivalent to saying that the problem of deciding whether $I$ admits a minimum-regret stable matching with at most $k$ $r$th choices is NP-complete. To establish our result, we will show that for an instance $I$ with complete preference lists, $I$ admits a minimum-regret stable matching $M$ with at most $k$ $r$th choices in $p(M)$ if and only if $I$ admits a minimum-regret stable partition $\Pi$ with at most $2k$ $r$th choices in $p(\Pi)$.

    Suppose that $I$ admits a minimum-regret stable matching $M=\{\{a_{i_1},a_{i_2}\}\dots \{a_{i_{n-1}}, a_{i_{n}}\}\}$ with at most $k$ $r$th choices, i.e. at most $k$ agents are assigned a partner that is ranked in the $r$th position of their preference list. Then, immediately, $M$ has an associated minimum-regret stable partition $\Pi = (a_{i_1}\; a_{i_2})\dots(a_{i_{n-1}}\; a_{i_{n}})$, in which at most $k$ agents have a predecessor that is ranked in their $r$th position. As all cycles are transpositions, their predecessors are also their successors, so in the associated profile we will have $p_r(\Pi)\leq 2k$ as desired.

    For the other direction, suppose $I$ admits a minimum-regret stable partition $\Pi$ with $p_r(\Pi)\leq 2k$. By Lemma \ref{lemma:lastchoicereduced}, we can furthermore suppose that $\Pi$ is reduced, because if $\Pi$ is not reduced, then we can break its even cycles into a stable collection of transpositions in a way that does not increase the number or index of worst-choice assignments. Because we assume that $I$ is solvable, by Lemma \ref{lemma:minregret}, $\Pi$ corresponds to a minimum regret stable matching $M$ of $I$, and because all cycles in $\Pi$ are transpositions, we know that $p_r(\Pi) = 2p_r(M)$, as desired.
\end{proof}

\begin{theorem}
\label{theorem:rthchoicehardness}
    Let $I$ be an {\sc sr} instance with regret $r$ and let $k$ be an integer. {\sc RM-Dec-SR-SP}, the problem of deciding whether $I$ admits a minimum-regret stable partition with at most $k$ $r$th choices, is NP-complete, regardless of whether $I$ is guaranteed to be solvable or unsolvable.
\end{theorem}
\begin{proof}
    If $I$ is solvable, the result follows immediately from Lemma \ref{lemma:rthchoicecomplete}. 
    
    If $I$ is unsolvable, we can use the same construction as in Theorem \ref{theorem:fchardness} to show that NP-completeness follows from the solvable case.
\end{proof}

This immediately lets us establish hardness results for generous stable partitions.

\begin{theorem}
\label{theorem:gencompleteness}
    {\sc Gen-Dec-SR-SP} is NP-complete, regardless of whether $I$ is guaranteed to be solvable or unsolvable.
\end{theorem}
\begin{proof}
    Note that {\sc Gen-Dec-SR-SP} is in NP, as given an instance $I$, vector $\sigma$, and partition $\Pi$, we can check that $\Pi$ is a stable partition of $I$ and check that $\sigma^{rev}\succeq p^{rev}(\Pi)$, all in linear time.
    
    We showed above in Theorem \ref{theorem:rthchoicehardness} that {\sc RM-Dec-SR-SP} is NP-complete for solvable and for unsolvable instances. Now given a problem instance $(I, k)$ to {\sc RM-Dec-SR-SP}, $(I, k)$ is a yes-instance if and only if $(I,\sigma)$ is a yes-instance to {\sc Gen-Dec-SR-SP}, where $r$ is the minimum regret of any stable partition of $I$, $\sigma_r = k$, $\sigma_i=n$ for $i<r$ and $\sigma_i=0$ for $i>r$. 
    
    In this construction, both problems ask whether there exists a stable partition $\Pi$ of $I$ where at most $k$ agents are assigned their $r$th choice and no agent is assigned to anyone worse than their $r$th choice, as any stable partition satisfies $p_i(\Pi) < \sigma_i$ for $i<r$, so the result follows.
\end{proof}

Finally, we consider minimising the profile-based (egalitarian) cost of a stable partition. 

Again, the following lemma lets us reduce the search to reduced stable partitions.

\begin{lemma}
\label{lemma:costreduced}
    Let $C=(a_{i_1} \; a_{i_2} \dots a_{i_k})$ be an even cycle with cost $c(C)$. Then both collections of transpositions $C_1 = (a_{i_1}\; a_{i_2}) \dots (a_{i_{k-1}} \; a_{i_{k}})$ and $C_2 = (a_{i_1}\; a_{i_{k}}) \dots (a_{i_{k-2}} \; a_{i_{k-1}})$ are stable and either $c(C_1)\leq c(C)$ or $c(C_2)\leq c(C)$.
\end{lemma}
\begin{proof}
    By Lemma \ref{lemma:atleasttwo}, both $C_1$ and $C_2$ are stable. Now let $c_1,c_2$ be the costs of $C_1, C_2$, respectively. Let rank$_s(C)=(sr_1 \; sr_2 \dots sr_k)$ be the vector of successor ranks for agents $i_1\dots i_k$, i.e., $sr_j$ is the rank of $C(a_{i_j})$ on $a_{i_j}$'s preference list and define rank$_p(C)$ the same for predecessors. Then $c(C)=$ sum(rank$_s(C)$ + rank$_p(C)$), and similarly for $C_1, C_2$, by definition. However, by construction, we have that rank$_s(C_1)=(sr_1 \; pr_2 \; sr_3 \dots pr_k)=$ rank$_p(C_1)$ and rank$_s(C_2)=(pr_1 \; sr_2 \; pr_3 \dots sr_k)=$ rank$_p(C_2)$. Then $c_1+c_2=2c(C)$, so the result follows.
\end{proof}

This helps us in establishing the intractability of the problem.

\begin{lemma}
\label{lemma:egalcomplete}
    {\sc Egal-Dec-SR-SP} is NP-complete, regardless of whether the given instance is guaranteed to be solvable or unsolvable. 
\end{lemma}
\begin{proof}
    Note that {\sc Egal-Dec-SR-SP} is in NP, as given instance $I$, integer $k$, and partition $\Pi$, we can check whether $\Pi$ is a stable partition of $I$ and verify that it has cost at most $k$, all in linear time.

    \citet{federegal} proved that {\sc Egal-Dec-SR-SM}, the problem of deciding whether $I$ admits a stable matching with cost at most $k$, is NP-complete, even if $I$ is solvable. To establish our result, we will show that for an instance $I$ with complete preference lists, $I$ admits a stable matching $M$ with cost at most $k$ if and only if $I$ admits a stable partition $\Pi$ with cost at most $k$.

    By solvability of $I$ and Theorem \ref{theorem:correspondence}, any stable matching $M=\{\{a_{i_1},a_{i_2}\}\dots \{a_{i_{n-1}}, a_{i_{n}}\}\}$ corresponds to a reduced stable partition $\Pi = (a_{i_1}\; a_{i_2})\dots(a_{i_{n-1}}\; a_{i_{n}})$ and vice versa. Furthermore, by definition of cost, $c(\Pi)=c(M)$. Also, by Lemma \ref{lemma:costreduced}, there is no gap in minimum cost between reduced stable partitions and non-reduced stable partitions, and we know that we can turn any non-reduced stable partition into a reduced stable partition without an increase in cost. Thus, the result follows.
\end{proof}

\begin{theorem}
    \label{theorem:egalhardness}
    Let $I$ be an {\sc sr} instance and $k$ be an integer. {\sc Egal-Dec-SR-SP}, the problem of deciding whether $I$ admits a stable partition with egalitarian cost at most $k$, is NP-complete, regardless of whether $I$ is guaranteed to be solvable or unsolvable.
\end{theorem}
\begin{proof}
    If $I$ is solvable, the result follows immediately from Lemma \ref{lemma:egalcomplete}. If $I$ is unsolvable, we can use the same construction as in Theorem \ref{theorem:fchardness} to show that NP-completeness follows from the solvable case, as the cost of the added 3-cycles in the reduction is fixed and can be added onto the parameter $k$.
\end{proof}

Note that all NP-hardness arguments above assume a solvable (sub)instance that is computationally difficult in the considered context. However, the smaller the sub-instance, the easier to compute an optimal solution, as there is no choice involved in the odd cycles. In the extreme case, all agents are part of odd cycles, which means that there exists a unique stable partition for the instance, which is automatically optimal regardless of the objective.

\subsection{Approximability of the NP-hard Problems}
\label{section:approxofproblems}

After seeing how profile optimality transfers from matchings to partitions in terms of computational complexity, we can similarly consider whether approximation methods carry over. We have shown the NP-hardness of computing egalitarian, first-choice maximal, rank-maximal, regret-minimal, and generous stable partitions. However, there is no known measure of approximability for rank-maximal or generous profiles, even for stable matchings, so it is unclear what an approximation could look like for stable partitions. We postpone this for future work and focus on the first-choice maximal, regret-minimal, and egalitarian problems.

\begin{theorem}
\label{theorem:fcinapprox}
    The problem of finding a first-choice maximal stable partition $\Pi$ of some {\sc sr} instance $I$ does not admit a polynomial-time approximation algorithm with any constant-factor performance guarantee, regardless of whether $I$ is guaranteed to be solvable or unsolvable, unless P$=$NP. 
\end{theorem}
\begin{proof}
    \citet{simola2021profilebased} proved that there does not exist any constant factor approximation algorithm for the first-choice maximal stable matching problem, unless P$=$NP. Thus, for this proof, it suffices to show that if there were a polynomial-time approximation algorithm for regret-minimal stable partitions with any constant performance guarantee $c$, then we could also solve the respective stable matching problem.

    Indeed, let $I$ be any solvable {\sc sr} instance and let $opt$ be the maximum number of first choices achieved in any stable matching of $I$. We have previously shown in the proof of Lemma \ref{lemma:firstchoice} that $2opt$ is the maximum number of first choices achieved in the profile of any stable partition of $I$ and that it can be achieved by some reduced stable partition. Thus, if we could approximate a first-choice maximal stable partition problem within $c$, then we could find a reduced stable partition $\Pi$ with at least $\frac{2opt}{c}$ first choices in the profile, corresponding to $\frac{opt}{c}$ first choices in the corresponding stable matching profile, which is a contradiction to the fact that there is no constant factor approximation algorithm for the stable matching version of the problem.

    Furthermore, if we had a $c$-approximation algorithm for unsolvable instances, then we could transform any solvable instance $I$ with $n$ agents into an unsolvable one by adding two forced odd cycles $(a_{n+1} \; a_{n+2} \; a_{n+3})(a_{n+4} \; a_{n+5} \; a_{n+6})$ similar to the proof of Theorem \ref{theorem:fchardness} using the preference list construction from Example \ref{table:firstchoice3cycles}, where $a_{n+i}=b_i$. Clearly, the preferences admit the desired 3-cycles $(b_1\;b_2\;b_3)(b_4\;b_5\;b_6)$ as stability conditions T1 and T2 are satisfied. Furthermore, as the cycles are of odd length, we can immediately conclude that they are unique and must thus be part of any stable partition. The cycles do not satisfy any additional first choices, which could affect the correctness of the reduction to prove the inapproximability of first-choice maximal stable partitions of unsolvable instances within any constant factor.
\end{proof}

For the other problems, we can use some known results about approximating special stable matchings. 

\begin{definition}
    Let $I$ be an instance of {\sc sr} and let $c_i: A\rightarrow\mathbb{R}$ be a real-valued cost function that indicates the cost $c_i(a_j)$ incurred by $a_i$ when being matched to $a_j$. The {\sc Optimal sr} problem is the problem of finding a stable matching $M$ of $I$ (if one exists) such that $\sum_{a_i\in A(M)} c_{i}(M(a_i))$ is minimal among all stable matchings, where $A(M)$ is the set of agents matched in $M$.
\end{definition}

\citet{teo97} showed that if the cost function furthermore satisfies that for each agent $a_{k}$ with complete ordinal preferences $a_{k_1}\succ_k \dots \succ_k a_{k_{n-1}}$, there exists an agent $a_{k_l}$ such that $c_k(a_{k_{m-1}})\geq c_k(a_{k_{m}})$ for all $1<m\leq l$ and $c_k(a_{k_{m'}})\leq c_k(a_{k_{m'+1}})$ for all $l\leq m'< n-1$, then {\sc Optimal sr} is approximable within a factor of 2 using an integer programming approach computable in polynomial time. 

Note that for first-choice maximal stable partitions, it is not possible to apply this approximation method, as there is no suitable cost function that satisfies the additional criteria. However, \citet{simola2021profilebased} previously showed how the 2-approximation of {\sc Optimal sr} can be used to approximate regret-minimal stable matchings, and we showed that this is also true for stable partitions.

\begin{theorem}
\label{theorem:rmapprox}
    The problem of finding a regret-minimal stable partition $\Pi$ of some {\sc sr} instance $I$ admits a polynomial-time approximation algorithm with a performance guarantee of 2. 
\end{theorem}
\begin{proof}
    Again, we know that because all odd cycles are invariant and can be found in polynomial time, we show that we can 2-approximate the even cycles, which actually gives an overall performance guarantee significantly better than 2 if the number of agents in odd cycles is large compared to the number of agents in the instance. 
    
    Let $\Pi_p$ be a minimum regret stable partition of $I$ with profile $p$, which always exists, and we have shown to be computable in linear time, and let $r$ be the minimum regret. If $r=1$ or $p_r=p_r\vert_{\mathcal{O}_I}$, i.e. all profile entries at position $r$ are caused by half-matchings in odd cycles, there is no variability, and we found an optimal solution. Otherwise, we showed that $\Pi_p=M'\mathcal{O}_I$ for a minimum-regret stable matching $M$ of $I_P$ with dummy agents removed. Some match in $M$ causes an entry at position $r$ in the profile of $\Pi$, none of which can be with a dummy agent, as they are all matched to each other and rank their partner first in their preference list. Notice that agents rank their dummy partners first because we are using the padding method described in Section \ref{section:paddingargument}, not the one described in the proof of Theorem \ref{theorem:fchardness}, as we start with the polynomial-time algorithm to find a minimum-regret stable partition. We define the cost function 
    \begin{equation*}
    c_i(a_j) = \begin{cases}
        1 &\text{rank}_i(a_j)\geq r\\
        0 &\text{otherwise}.
    \end{cases}
    \end{equation*}
    Then, $c$ satisfies the condition by \citet{teo97} as preference lists correspond to a cost sequence $0,\dots,0,1,\dots,1$ with $r-1$ zeros and $n-r$ ones. Furthermore, as we know that every agent in $S_P$ can be matched to some other agent of rank at least as good as $r$, we can truncate the preference lists of every agent at rank $r+1$. Then $\sum_{a_i\in A_M} c_i(M(a_i))$ is exactly the number of agents in $I_P$ that are matched to their $r$th choice. Now, let $\Pi_{opt}$ be an optimal stable partition, then we can apply the 2-approximation algorithm for {\sc Optimal SR} with cost function $c$ to $S_P$ to get a matching $M''$ with dummy agents removed such that $p_r(M''\mathcal{O}_I) = p_r(M'')+p_r(\mathcal{O}_I) \leq 2 (p_r(M_{opt})+p_r(\mathcal{O}_I))=2p_r(\Pi_{opt})$ as desired.
\end{proof}

Similarly, we showed that there also exists a suitable cost function to approximate egalitarian stable partitions using this method.

\begin{theorem}
\label{theorem:egalapprox}
    The problem of finding an egalitarian stable partition $\Pi$ of some {\sc sr} instance $I$ admits a polynomial-time approximation algorithm with a performance guarantee of 2. 
\end{theorem}
\begin{proof}
    We will show that we can define a suitable cost function $c$ such that when applying the 2-approximation algorithm to the {\sc Optimal SR} problem to the solvable instance $I_P$ derived from $I$ by the padding method, we find a solution to the egalitarian stable partition problem with a performance guarantee of 2. Even more, we know that because all odd cycles are invariant and can be found in polynomial time, we show that we can 2-approximate the even cycles, which actually gives an overall performance guarantee significantly better than 2 if the number of agents in odd cycles is large compared to the number of agents in the instance. Indeed, the natural cost function $c_i(a_j)=$ rank$_i(a_j)$ works to approximate an egalitarian stable matching within 2, as $c_k(a_{k_i})$ is strictly increasing for $i$ from 1 to $n-1$. Also, the solution $M$ necessarily minimises the cost $c(M)=\sum_{a_i\in A_M} c_i(M(a_i))$ as required.

    Now we want to show that $\Pi=M\mathcal{O}_I$, where $M$ is a 2-approximation of an egalitarian stable matching of $I_P$ with dummy agents removed and $\mathcal{O}_I$ are the invariant odd cycles of $I$, is a 2-approximation of an optimal egalitarian stable partition $\Pi_{opt} = M'\mathcal{O}_I$ of $I$, where $M'$ is an egalitarian stable matching of $I_P$ as previously shown.  Recall that $c(\Pi) = \frac{1}{2} \sum_{1\leq i< n} p_i\cdot i =\frac{1}{2}\sum_{a_i}(\text{rank}_i(\Pi(a_i))+\text{rank}_i(\Pi^{-1}(a_i)))$. Thus, $c(\Pi) = c(M)+c(\mathcal{O}_I)$ and similarly for $c(\Pi_{opt})$, such that $c(\Pi)>2c(\Pi_{opt})$ implies $c(M)>2c(M')$, a contradiction. Therefore, the method gives a 2-approximation for the egalitarian stable partition problem and can be computed in polynomial time.
\end{proof}

These results are consistent with the approximability results of stable matchings, which fits the theme of compatibility and close correspondence throughout this research.

\section{Experimental Results}
\label{section:experiments}
\captionsetup[table]{name=Table}

To obtain an intuition for the structures studied in this paper, we implemented the new algorithms and ran experiments on random {\sc sr} instances, measuring various properties of the solutions. Specifically, we were interested in the number of reduced and non-reduced stable partitions and the cycles contained within them to gain an understanding of the structural properties of such partitions. As the odd cycles are clearly of particular importance for stable partitions and for the solvability of {\sc sr} instances altogether, we also aimed to gain a deeper understanding of the likelihood of existence and lengths of the odd cycles and the number of agents contained within them, extending empirical results provided by \citet{mertens15small,mertens15random}.

For our experiments, we generated synthetic {\sc sr} instances with preferences chosen uniformly at random, with the number of agents $n$ ranging from 2 to 500. For each $n$, we generated 10,000 seeded random instances over which the following results are averaged. For all instances, we computed a stable partition using a custom implementation of the Tan-Hsueh algorithm \citep{tanhsueh} and captured the numbers and lengths of cycles observed. Furthermore, for a subset of these instances with $n$ between 10 and 500, we enumerated their reduced and non-reduced stable cycles and stable partitions using the novel algorithms presented in Section \ref{section:cyclesandpartitions}. All implementations were written in Python, and all computations were performed on the {\tt fatanode} cluster.\footnote{See \href{https://ciaranm.github.io/fatanodes.html}{https://ciaranm.github.io/fatanodes.html} for technical specifications.} We also implemented the Tan-Hsueh algorithm in Java and made it accessible through the \emph{Matching Algorithm Toolkit Web Application}.\footnote{See \href{https://matwa.optimalmatching.com/}{https://matwa.optimalmatching.com/} for the web application and \citep{matwa} for an accompanying demo paper.}

Although measuring the runtime of these algorithms was not an objective in these experiments, it might be interesting to know that all algorithms, even those with exponential running time bounds, took, on average, less than 1 second of CPU time for instances with $n\leq 100$, but up to around one hour on some outlier instances with $n=500$.

We denote the estimated value of $P_n$ by $\hat{P}_n$, that is, the proportion of generated instances with $n$ agents that were found to be solvable. Table \ref{table:numstructures} shows $\hat{P}_n$ and the average number of reduced stable partitions ($RP$) and all stable partitions ($P$), as well as the number of reduced stable cycles ($RSC$) and all stable cycles ($SC$) contained within them. With regards to $\hat{P}_n$, our results closely match previously known estimates for uniform random instances due to \citet{pittelirving94,mertens05}, with $\hat{P}_n$ decreasing from 100\% for $n=2$ down to 64.50\% for $n=100$ and 45.00\% for $n=500$. With regards to the other properties, it is interesting to see that the expected number of stable partitions remains small, averaging fewer than 7 stable partitions and fewer than 3 reduced stable partitions even for $n=500$. Furthermore, the results suggest that the expected number of stable cycles grows slightly faster than that of reduced stable cycles, averaging a bit higher than the necessary minimum number of cycles, which is approximately $\frac{n}{2}$ (assuming most cycles are of length 2, as we will see later). However, the maximum number of stable partitions observed is 729 for an instance of size $n=500$, whereas the highest observed number of reduced stable partitions for any instance is 65, admitted by an instance with 300 agents. Interestingly, the maximum number of reduced stable cycles and stable cycles admitted by any instance are 315 and 323, respectively, which is a different $n=500$ instance than previously. Thus, outlier instances that admit many more such structures than average ones are present and observable.

\begin{table}[!htb]
    \centering
    \caption{Average Cardinality of $RP$, $P$, $RSC$ and $SC$ for different $n$.}
    \small
    \begin{tabular}{c c c c c c c c c c c c}
        \toprule
        & \textbf{10} & \textbf{20} & \textbf{30} & \textbf{40} & \textbf{60} & \textbf{80} & \textbf{100} & \textbf{200} & \textbf{300} & \textbf{400} & \textbf{500} \\
        \midrule
        $\hat{P_n}$ (\%) & 88.82 & 82.38 & 79.23 & 75.74 & 70.15 & 67.36 & 64.50 & 56.28 & 51.34 & 46.93 & 45.00 \\
        \midrule
        $\vert RP \vert$ & 1.38 & 1.63 & 1.76 & 1.87 & 1.97 & 2.06 & 2.13 & 2.43 & 2.63 & 2.70 & 2.83 \\
        $\vert P \vert$  & 1.76 & 2.34 & 2.67 & 2.94 & 3.26 & 3.54 & 3.83 & 4.96 & 5.80 & 6.10 & 6.72 \\
        \midrule
        $\vert RSC \vert$ & 5.79 & 11.38 & 16.64 & 21.90 & 32.03 & 42.22 & 52.24 & 102.61 & 152.83 & 202.59 & 252.80 \\
        $\vert SC \vert$ & 6.17 & 11.98 & 17.34 & 22.67 & 32.86 & 43.09 & 53.14 & 103.63 & 153.93 & 203.72 & 253.98 \\
        \bottomrule
    \end{tabular}
    \label{table:numstructures}
\end{table}

Figure \ref{fig:cyclesvn} considers only unsolvable instances and shows the expected number of stable cycles of certain odd lengths as the number of agents grows. It is clear that for $n=4$, every unsolvable instance admits exactly one stable cycle of length 1 (henceforth referred to as a 1-cycle) and one stable cycle of length 3 (3-cycle). However, it is interesting to see that while the likelihood of 1-cycles decreases quickly and is essentially 0 for $n=100$, the 3-cycles remain the most common odd cycles throughout. Furthermore, the expected numbers of cycles of length 5 or longer remain relatively low compared to that of 3-cycles, although we can see a clear hierarchy in the sense that (for $n$ even between 4 and 100) we should expect more 5-cycles than 7-cycles than 9-cycles, etc. The numbers remain so low though that it motivates the question of how many odd cycles can be expected altogether.

\begin{figure}[!htb]
    \centering
    \begin{tikzpicture}
        \begin{axis}[
            width=.7\textwidth,
            height=6cm,
            xlabel={Number of Agents ($n$)},
            ylabel={Number of Cycles},
            legend cell align={left},
            legend style={
                at={(1.05,1)}, % Adjusts the legend position
                anchor=north west,
                column sep=1ex, % Space between legend columns
            },
            grid=both,  % Add both major and minor gridlines
            grid style={dashed, gray!30},  % Subtle grid lines
            tick style={black},
            cycle list name=color list, % Define color cycle
            every axis plot/.append style={thick},  % Make all lines thick
        ]
        
        % Series 1
        \addplot[acmDarkBlue, mark=*] table [x=n, y=1-cycles] {data/data.txt};
        \addlegendentry{Cycles of Length 1}

        % Series 2
        \addplot[acmGreen, mark=square*] table [x=n, y=3-cycles] {data/data.txt};
        \addlegendentry{Cycles of Length 3}

        % Series 3
        \addplot[acmPink, mark=triangle*] table [x=n, y=5-cycles] {data/data.txt};
        \addlegendentry{Cycles of Length 5}

        % Series 4
        \addplot[acmOrange, mark=diamond*] table [x=n, y=7-cycles] {data/data.txt};
        \addlegendentry{Cycles of Length 7}

        % Series 5
        \addplot[acmYellow, mark=o] table [x=n, y=9-cycles] {data/data.txt};
        \addlegendentry{Cycles of Length 9}

        % Series 6
        \addplot[acmLightBlue, mark=star] table [x=n, y=11-cycles] {data/data.txt};
        \addlegendentry{Cycles of Length 11}

        \end{axis}
    \end{tikzpicture}
    \caption{Expected Number of Odd-Length Cycles in Unsolvable Instances}
    \label{fig:cyclesvn}
    \Description{A detailed plot of the average number of cycles of specific sizes for the instances we studied. It can be seen that cycles of length 1 become increasingly rare as the number of agents is increased. It can also be seen that the average number of cycles of a fixed odd length at most 11 is below 1 for 100 agents.}
\end{figure}

Table \ref{table:oddvagents} gives an indication of the numbers and lengths of odd cycles to expect in unsolvable instances. We can see that, indeed and as expected based on the previous figure, in expectation, an instance only admits a very small number of odd cycles, and this parameter grows very slowly. Notice that any even-sized unsolvable instance must admit at least two odd cycles, so the results indicate that most instances do not admit any more than those two. This intuition is supported when looking at the maximum number of odd cycles admitted by any instance in the dataset, which is 2, 4 and 6 for $n\in \{10,100,500\}$, respectively. Furthermore, the table shows that the average number of agents in odd cycles does grow (not steeply, but significantly faster than the expected number of odd cycles), and with it the average length of these cycles. Indeed, the range of values observed here is much wider, with as many as 92 agents contained in odd cycles observed for an instance with $n=100$.

\begin{table}[!htb]
    \centering
    \caption{Odd-Length Cycle Properties for different $n$}
    \begin{tabular}{c c c c}
        \toprule
        & \textbf{10} & \textbf{100} & \textbf{500} \\
        \midrule
        \text{Average odd cycle length}   & 2.61 & 6.09 & 10.41 \\
        \text{Average number of odd cycles}       & 2.00 & 2.04 & 2.15 \\
        \text{Average number of agents in odd cycles} & 5.23 & 12.40 & 22.35 \\
        \bottomrule
    \end{tabular}
    \label{table:oddvagents}
\end{table}

Based on these preliminary experimental results, it would appear that, for small and moderate instances sizes (of up to 500 agents), the enumeration of stable partitions (both reduced and non-reduced) is often feasible in practice, and so even though we proved that the problems of finding egalitarian, first-choice maximal and regret-minimal stable partitions are NP-hard, computing such stable partitions can be feasible in many cases by simply iterating over the (small) set of all reduced stable partitions (we established in Section \ref{section:complexityofproblems} that there is no difference between the objective values achieved by optimal reduced and non-reduced stable partitions for these problems). Furthermore, the results suggest that while the average number of odd cycles is near-constant at 2, the average odd cycle length does grow. However, this motivates maximum stable matchings, computable in $O(n^2)$ time \citep{tan91_2}, further as a suitable solution concept in some settings, as only one agent from each odd cycle remains unmatched, and we expect this number to be small even for $n=500$. 

\section{Discussion and Open Problems}
\label{section:discussion}

In this paper, we have shown various new insights into the structure of stable partitions. Specifically, we showed that there is a close correspondence between reduced stable partitions of unsolvable instances and stable matchings of certain solvable sub-instances, and how even cycles can be constructed from transpositions and partial cycles, with the highlight being a result proving that a predecessor-successor pair is sufficient to find the unique stable cycle of even length longer than 2 (if one exists). Using these insights, we also provided the first algorithms to efficiently enumerate all reduced and non-reduced stable cycles and partitions of a given problem instance, which could be useful for a variety of cases. Furthermore, we showed how the stable matching notions of profile- and cost-optimality carry over to stable partitions and which of the problems are efficiently computable. Although we assumed complete preference lists throughout this work, the same results and algorithms should apply to {\sc sri} instances ({\sc sr} instances with incomplete lists). This bridges the gap between stable matchings, which do not always exist, and fractional matchings, which always exist but may not be useful in certain practical applications. 

For a comparative overview of the complexity results previously known and newly established, see Table \ref{table:complexity}. There, we consider maximum welfare stable fractional matchings as egalitarian for consistency. With regards to approximation bounds for the optimisation versions of the NP-complete decision problem variants above, Table \ref{table:approx} compares the existing results for stable matchings to the new results for stable partitions presented here.

\renewcommand{\arraystretch}{1.2}

\begin{table}[!htb]
    \centering
    \caption{Complexity classification of different {\sc sr} problems related to Stable Matchings (M), Stable Partitions (SP), and Fractional Matchings (FM). New contributions in colour. Not applicable or not a natural problem to study in grey. $n$ is the number of agents in the instance, $n_1$ of which are not in invariant odd cycles.}
    \small
    \begin{tabular}{c c c c}
        \toprule
        & \textbf{M} & \textbf{SP} & \textbf{FM} \\ 
        \midrule
        Always Exists & No & Yes & Yes \\ 
        Find Any (if exists) & $O(n^2)$ \citep{irving_sr} & $O(n^2)$ \citep{tan91_1} & $O(n^2)$ \citep{Caragiannis2019} \\ 
        \midrule\noalign{\vskip -2pt}
        Find $SP(I)$/$SC(I)$ & $O(n^3)$ \citep{federegalapprox} & \cellcolor[HTML]{b3ffb3}$O(n_1^2(n^2+n_1))$ (Thm \ref{theorem:findallsc}) & \cellcolor{lightgray} \\ 
        Find $FP(I)$/$FC(I)$ & $O(n^2)$ \citep{gusfield89} & \cellcolor[HTML]{b3ffb3}$O(n^2)$ (Thm \ref{theorem:findallfc}) & \cellcolor{lightgray} \\ 
        \midrule\noalign{\vskip -2pt}
        Find $S(I)$/$RP(I)$ & $O(\vert S(I)\vert n + n^2)$ \citep{federegalapprox} & \cellcolor[HTML]{b3ffb3}$O(\vert S(I_T)\vert n + n^2)$ (Thm \ref{theorem:findallrp}) & \cellcolor{lightgray} \\ 
        Find $P(I)$ & \cellcolor{lightgray} & \cellcolor[HTML]{b3ffb3}$O((\vert P(I)\vert + n_1)n_1n^2+n_1^3)$ (Thm \ref{theorem:enumpi}) & \cellcolor{lightgray} \\ 
        \midrule\noalign{\vskip -2pt}
        Min Regret & $O(n^2)$ \citep{gusfield89} & \cellcolor[HTML]{b3ffb3}$O(n^2)$ (Thm \ref{theorem:minregret}) & \cellcolor{lightgray} \\ 
        Egalitarian & NP-hard \citep{federegal} & \cellcolor[HTML]{ff9999}NP-hard (Thm \ref{theorem:egalhardness}) & NP-hard \citep{Caragiannis2019} \\ 
        First-Choice Max & NP-hard \citep{CooperPhD} & \cellcolor[HTML]{ff9999}NP-hard (Thm \ref{theorem:fchardness}) & \cellcolor{lightgray} \\ 
        Rank-Max & NP-hard \citep{CooperPhD} & \cellcolor[HTML]{ff9999}NP-hard (Thm \ref{theorem:rmcompleteness}) & \cellcolor{lightgray} \\ 
        Regret-Min & NP-hard \citep{CooperPhD} & \cellcolor[HTML]{ff9999}NP-hard (Thm \ref{theorem:rthchoicehardness}) & \cellcolor{lightgray} \\ 
        Generous & NP-hard \citep{CooperPhD} & \cellcolor[HTML]{ff9999}NP-hard (Thm \ref{theorem:gencompleteness}) & \cellcolor{lightgray} \\ 
        \bottomrule
    \end{tabular}
    \label{table:complexity}
\end{table}

\begin{table}[!htb]
    \centering
    \caption{Approximation Results for Different {\sc sr} Problems. New contributions in colour.}
    \small
    \begin{tabular}{c c c}
        \toprule
        & \textbf{Stable Matching} & \textbf{Stable Partition} \\ 
        \midrule\noalign{\vskip -2pt}
        Egalitarian & 2-approximable \citep{federegalapprox} & \cellcolor[HTML]{b3ffb3}2-approximable (Thm \ref{theorem:egalapprox}) \\ 
        Regret-Min & 2-approximable \citep{simola2021profilebased} & \cellcolor[HTML]{b3ffb3}2-approximable (Thm \ref{theorem:rmapprox}) \\ 
        First-Choice Max & Not approx within any $c\in\mathbb{R}$ \citep{simola2021profilebased} & \cellcolor[HTML]{ff9999}Not approx within any $c\in\mathbb{R}$ (Thm \ref{theorem:fcinapprox}) \\ 
        \bottomrule
    \end{tabular}
    \label{table:approx}
\end{table}

Overall, we have shown a close correspondence between stable matchings and stable partitions in every way, from instance and solution structure to the computational complexity and approximability of problem variants. Follow-up work in \citep{exp} greatly extends the experimental results presented in this paper, for example, by considering random preferences sampled from different statistical cultures and instances with many more agents. Furthermore, follow-up work in \citep{glitzner25sagt,glitzner2025unsolvabilitymanytomanynonbipartitestable} generalises the stable partition concept to the many-to-many extension of {\sc Stable Roommates}, referred to as the {\sc Stable Fixtures} problem, and proves algorithmic and structural properties about them which are similar to those previously established by Tan in Theorem \ref{thm:tan} and to those we established in this paper.

However, there are still many open problems related to stable partitions and other topics of this paper. One direction for future work is to improve on the complexity results presented here, for example by investigating whether a technique similar to the {\sc 2-sat} reduction by \citet{federegalapprox} in the enumeration of $SP(I)$ and $S(I)$ or the transformations between non-bipartite to bipartite instances presented by \citet{deanmunshi10} can speed up the enumeration of $SC(I)$ and $P(I)$. Another algorithmic question is whether we can find a stable partition with the maximum number of agents in even cycles efficiently. One could also explore stable fractional matchings more, looking into fast enumeration of such matchings and partnerships, or adapting optimality criteria similar to the ones formulated here. 

Currently, it is not clear what a suitable definition for an approximation ratio could look like for the rank-maximal and generous-profile problems, neither for matchings nor for partitions. On the structural side, it would also be interesting to study whether stable half-matchings form a semi-lattice, as is the case for stable matchings \citep{gusfield89}. There might also be a closer relationship than currently discovered between the structure of stable partitions and approximate or exact solutions to computationally difficult problems, such as the almost-stable matching problem, where some structural results presented here might help. 

Finally, it could be investigated how stable partitions can be useful in the presence of ties in the preference lists, for example, in approximating large weakly stable matchings.

\paragraph{Software and Data Availability} 
The code to generate the same instances that we used in our experiments, and to replicate our results, is openly available and linked in reference \citep{experimentsCode}. For instructions on how to generate the instances and run the algorithms, we refer to the {\tt README.md} file. Note that the code offers many functionalities that are not used in this paper. A separate paper \citep{exp} that also makes use of this code and focuses on the experimental aspects of random {\sc sr} instances is currently in submission.

\begin{acks}
We would like to thank the anonymous MATCH-UP, SAGT, and TEAC reviewers for their helpful comments and suggestions. Frederik Glitzner is supported by a Minerva Scholarship from the School of Computing Science, University of Glasgow. David Manlove was supported by the Engineering and Physical Sciences Research Council, grant number EP/X013618/1. For the purpose of open access, the authors have applied a Creative Commons Attribution (CC BY) licence to any Author Accepted Manuscript version arising from this submission.
\end{acks}

\bibliographystyle{ACM-Reference-Format}
\bibliography{references}


\end{document}